%% file: voute.tex
\documentclass{article}
\usepackage[hyphens]{url}
\usepackage{algorithm}
\usepackage{algorithmicx}
\usepackage{algcompatible}
\usepackage{graphicx}
\usepackage{amsfonts}
\usepackage{amsmath}
\usepackage{amsthm}
\usepackage{enumerate}
\usepackage{multirow}
\usepackage{subfig}
\usepackage{eqparbox}
\usepackage[noadjust]{cite}
\usepackage{color}

\newtheorem{theorem}{Theorem}[section]
\newtheorem{lemma}[theorem]{Lemma}

\newtheorem{cor}[theorem]{Corollary}
\newtheorem{definition}[theorem]{Definition}

\title{VOUTE-Virtual Overlays Using Tree Embeddings}
\author{Stefanie Roos, Martin Beck, Thorsten Strufe \\
TU Dresden \\
\{stefanie.roos,martin.beck1,thorsten.strufe\}@tu.dresden.de}

\date{}

\newcommand{\E}{\mathbb{E}}
\newcommand{\calO}{\mathcal{O}} 

\newcommand{\commIt}[1] {\STATEx \COMMENT{\textit{#1}}}
\newcommand{\commItLine}[1] {\COMMENT{\textit{#1}}}

\newcommand{\routeN}{\mathbf{R}_{node}}
 \newcommand{\find}{\mathbf{R}_{content}}
\newcommand{\stab}{\mathbf{S}}
\newcommand{\ano}{\mathbf{AdGen}_{node}}
\newcommand{\anoC}{\mathbf{AdGen}_{content}}

\newcommand{\route}{\mathbf{R}}

\newcommand{\s}{s}
\newcommand{\e}{e}

\newcommand{\X}{\mathbf{X}}
\newcommand{\dist}{\delta_X}
\newcommand{\id}{\mathit{id}}

\newcommand{\keys}{\mathbb{K}}

\newcommand{\distIndex}[1]{\delta_{#1}}

\newcommand{\Y}{\mathbf{Y}}

\newcommand{\pre}{cord}

\newcommand{\xor}{\oplus}
\newcommand{\lengthAd}{L}

\newcommand{\keysPart}{\tilde{\keys}}

\newcommand{\trees}{\gamma}
 
\newcommand{\keymac}{\mathbb{K}_{MAC}} 
\newcommand{\Q}{Q} 
\newcommand{\routeTD}{\mathbf{R}^{TD}}
\newcommand{\routeRAPCPL}{\mathbf{R}^{CPL}}
\newcommand{\routePPPCPL}{\mathbf{R}^{PPP}}

\begin{document}
\maketitle
\begin{abstract}
Friend-to-friend (F2F) overlays, which restrict direct communication to mutually trusted parties, are a promising substrate for privacy-preserving communication due to their inherent membership-concealment and Sybil-resistance.
Yet, existing F2F overlays suffer from a low performance, are vulnerable to denial-of-service attacks, or fail to provide anonymity.
In particular, greedy embeddings allow highly efficient communication in arbitrary connectivity-restricted overlays but require communicating parties to reveal their identity.
In this paper, we present a privacy-preserving routing scheme for greedy embeddings based on anonymous return addresses rather than identifying node coordinates.
We prove that the presented algorithm are highly scalalbe, with regard to the complexity of both the routing and the stabilization protocols. 
Furthermore, we show that the return addresses provide plausible deniability for both sender and receiver. 
We further enhance the routing's resilience by using multiple embeddings and propose a method for efficient content addressing.
Our simulation study on real-world data indicates that our approach is highly efficient and effectively mitigates failures as well as powerful denial-of-service attacks.
\end{abstract}

\input{intro}

\input{related}

\input{req}

\input{pre-embeddings}

\input{design-trees}
\input{design-embedding}
\input{design-returnAdd}

\input{design-content}
\input{eval-efficiency}

\input{eval-resilience}

\input{eval-anonymity}

\input{conc}

\bibliographystyle{unsrt}
\bibliography{diss}

\input{appendix}
\end{document}

%% file: intro.tex
\section{Introduction}

Anonymous and censorship-resistant communication is essential for providing freedom of speech.
In the last years, threats to this essential human right have emerged in western countries as well,
 mainly in the form of self-censorship caused by the fear of seemingly private communication being recorded \footnote{\tiny{\url{http://www.theguardian.com/commentisfree/2013/jun/17/chilling-effect-nsa-surveillance-internet}}}.
 Due to the natural vulnerability of publicly known servers to sabotage and corruption, completely distributed solutions for anonymous communication and content distribution are required.
 However, the openness of distributed systems presents a vulnerability, enabling attackers to infiltrate the system with 
 a large number of forged participants, as can seen e.g., from attacks on the Tor \cite{dingledine2004tor} network  in 2014 \footnote{\tiny{\url{https://blog.torproject.org/blog/ tor-security-advisory-relay-early-traffic-confirmation-attack}}}. 
 
 F2F overlays circumvent the problem of connecting to permanently changing strangers by restricting connectivity to participants sharing a mutual real-world trust relationship.
Hence, adversaries need to resort to social engineering attacks for infiltration of the network.
However, large-scale privacy-preserving communication in F2F overlays requires additional measures to achieve anonymity, failure and attack resilience, and efficiency. 
Multiple studies have shown that deployed F2F overlays such as Freenet \cite{clarke2010private} are highly inefficient and vulnerable to attacks \cite{vasserman2009membership, evans2007routing}.
 Virtual overlays have been proposed as an efficient alternative \cite{vasserman2009membership,mittal2012x}, but recent work has shown that they inherently require unacceptable high stabilization costs \cite{roos2015impossibility}. 

A potential solution is presented by greedy network embeddings such as \cite{kleinberg2007geographic,herzen2011scalable}.
Greedy embeddings allow for highly efficient greedy routing in arbitrary connectivity-restricted overlays. For this purpose, they first construct a spanning tree of the network and then assign coordinates based on a node's position in the spanning tree.
However, a participant can only contact a non-trusted contact when knowing its coordinate in the network.
Though only direct neighbors can directly map the embedding coordinate to a real-world identity, arbitrary participants can reconstruct the social graph based on the revealed coordinates. 
Participants can then easily be identified from the social graph structure \cite{narayanan2009anonymizing} and correlated with their activities in the network due to the coordinate acting as a persistent pseudonym.
In this manner, governmental and commercial institutions as well as curious strangers can track individual or all users and establish detailed profiles of their behavior.
Possibly even their opinions and interests, published unencrypted in a supposedly anonymous manner, are revealed to the adversary. 
Thus, network embeddings in their unaltered form fail to provide receiver anonymity.

For utilizing the high efficiency of network embedding, our first requirement is a  modified addressing and routing protocol that provides both efficiency and (receiver) anonymity. 
Second, due to the fragility of spanning trees in the presence of network dynamics and attacks, the resilience of the embedding to both failures and denial-of-service attacks needs to be drastically increased. 
Rather than only dropping messages, we assume that the adversary first strategically sabotages the embedding algorithm to maximize the impact of its censorship. 
Third, efficient content storage and retrieval requires the existence of a suitable content addressing scheme for network embeddings.

Our solution addresses the above problems by i) introducing anonymous return addresses to provide receiver anonymity, ii) constructing multiple embeddings and using backtracking during routing to increase the resilience, and iii) utilizing the network embedding to provide a routing protocol for a virtual overlay, thus avoiding the enormous stabilization costs of previous virtual overlays.

Our embedding algorithm assigns coordinates in the form of vectors of $b$-bit strings, so that nodes in the same subtree of the spanning tree share the same prefix, similar to the PIE embedding \cite{herzen2011scalable}.  
Rather than revealing the coordinate, the receiver then generates an anonymous return address by applying a hash cascade to the elements of the coordinate vector salted with a random seed.
After publishing the return address and the seed, the receiver can be contacted efficiently without revealing any information that is not required for routing. Furthermore, a node can publish several anonymous return addresses by varying the seed,
In this manner, it can construct distinct pseudonyms for distinct contexts, e.g., one pseudonym for each forum discussion it participates in. 
The revealed information can be further reduced by applying an additional layer of encryption at the price of a reduced efficiency.

By routing in multiple embeddings, we aim to increase the probability of finding a route despite the disruption of routes in some but not all embeddings.
We propose a purely local algorithm for the construction of multiple spanning trees of highly different structures to provide largely node-independent routes.
In addition, backtracking and a modified distance are integrated into the routing to further improve resilience and avoid congestion.

We evaluate our solution both by a formal security analysis and an extensive simulation study.
In the security analysis, we prove a receiver can never be uniquely identified from a return address.
Our simulation study indicates that our scheme is highly efficient compared existing approaches in terms of the number of messages required for routing.
Furthermore, the resilience is greatly improved. In fact, the routing terminates successfully in the vast majority of cases despite the presence of node failures or powerful attackers, which manipulate the embedding and routing as well as forge connections to honest nodes.

%% file: related.tex
\section{Related Work}
\label{sec:related}

Here, we describe the state-of-the-art with regard to routing and content discovery in F2F overlays. 

The common characteristics of all F2F overlays are i) the restriction of connections to trusted parties, ii) \emph{hop-by-hop anonymization}, i.e., the transfer of messages via a path of trusted nodes that rewrite the source tag of the message to point to themselves and apply probabilistic delays before forwarding a message, iii) encryption of all communication.
In the following, we present existing approaches, categorizing them according to their routing methodology in unstructured overlays, virtual overlays, and  network embeddings.
Routing is applied to either discover nodes based on network coordinates or, more commonly, content based on a content key or description.

\emph{Unstructured approaches} utilize \emph{flooding}, e.g. in Turtle \cite{popescu2006safe},  or \emph{probabilistic forwarding} e.g. in OneSwarm \cite{isdal2010privacy}. 
GnuNet attempts to combine random walks with deterministic routing \cite{evans2011r5n}.
These overlays focus locating content rather than individual nodes.
Due to the replication of content, the content can indeed be located, but efficient communication between two uniquely defined entities is not possible.


\emph{Virtual overlays} address the problem of establishing an overlay despite the restricted connectivity by replacing overlay links with tunnels of trusted nodes.
So, efficient tunnel discovery and maintenance is a main concern given the inherent network dynamics:
Vasserman et al. \cite{vasserman2009membership} suggest flooding the network for discovering adequate overlay neighbors, thus creating a large overhead.
In contrast, X-Vine leverages the overlay routing by concatenating previously existing tunnels to a new one, thus entailing a increase of the average tunnel length and hence routing costs over time \cite{mittal2012x}. 
Indeed, without an additional routing protocol in the underlying F2F overlay, 
efficient maintenance and efficient routing are inherently mutually exclusive in a virtual overlay \cite{roos2015impossibility}. 

In contrast, \emph{network embeddings} assign coordinates that allow efficient routing to nodes.
For example, the F2F mode of Freenet relies on a network embedding.
However, results indicate that the embedding is lacking both with regard to routing efficiency \cite{vasserman2009membership} and attack resilience \cite{evans2007routing}.
Hoefer et al. \cite{hofer2013greedy} propose highly efficient \emph{greedy embeddings}.
However, their approach reveals the identity of the communicating parties and fails to consider resilience.
Furthermore, their proposed scheme for content addressing maps the majority of content keys to the same central node.

In summary, network embeddings are the only existing approach providing a high efficiency at acceptable maintenance costs.
However, achieving receiver anonymity,resilience,  and suitably content addressing is an unsolved highly challenging problem.

%% file: req.tex
\section{Adversary Model}
\label{sec:req}

We aim to realize efficient F2F overlays making use of network embeddings but at the same time providing receiver anonymity, resilience, and content addressing.
Note that we do not consider sender anonymity because the problem of sender anonymity can easily be solved by starting the routing with a short random walk, as extensively analyzed for various anonymous look-up strategies for distributed hash tables (DHTs) (e.g., \cite{evans2011r5n,mclachlan2009scalable}).
In contrast, receiver anonymity is a challenging problem for network embeddings, because the coordinates acting as a node's pseudonym are essential for the routing process and hence for the efficient communication between arbitrary node pairs.
The term resilience is loosely defined. In general, a system is denoted resilience if an action is only slightly impaired by node failures or attacks.
Commonly, a system is judged to be resilient by comparison with others.

We consider two attack goals in our adversary model.
The first goal is to discover the identity of communicating parties, in particular the identity of the designated receiver of a message.
A second goal of the attacker is to block undesired communication using a so-called \emph{black hole attack} \cite{singh2006eclipse}, which could be applied in case an attacker fails to identify specific parties.
During such an attack, an adversary indiscriminately censors communication by first gaining a predominant position in the system and then dropping all received messages.
In addition, attacks on the availability, such as pollution, i.e., denial-of-service attacks by flooding the network with content and traffic, and eclipse attacks, i.e., censoring of specific content, present a thread for any P2P systems.
However, these attacks have been addressed in various publications (e.g., \cite{singh2006eclipse}), which can be applied to our contribution with few modifications. 
Hence, we do not consider them in our evaluation.

As for the attacker's capacities, we assume a local, active, internal, possibly colluding attacker, able to drop and manipulate messages it receives. 
The adversary can control one or several colluding nodes in the network  but is unable to observe the complete topology. In particular, we assume that an adversary cannot be certain that it knows all neighbors of a node, in other words, the complete circle of a user.We assume that this is hard, because it requires the adversary to i) be sure that he knows all contacts from  different social circles of a user, such as family, close friends, and colleagues, and ii) establish connections to all of them.
A global passive attacker is disregarded on the basis that steganographic techniques can be applied to hide the F2F traffic as suggested in e.g. \cite{mohajeri2012skypemorph}.
However, attackers are modeled as polynomial time adversaries, which are given a transcript of all own and public input, as well as all locally observed traffic.
Their computation power is therefore bounded by polynomial time algorithms, which prevents breaking computationally-secure cryptographic primitives.
We assume that an adversary can easily forge an arbitrary number of nodes, so called \emph{Sybils}. 
However, gaining connections and hence influence in a F2F overlay requires establishing real-world trust relationships.
Such \emph{social engineering attacks} are considered to be costly and difficult because they require long-term interaction between a human adversary and an honest participant.
Thus, the number of connections between honest nodes and forged participants can be assumed to be small,
More precisely, we assume that the number is logarithm with the network size, in agreement to previous work \cite{mittal2012x}.

%% file: pre-embeddings.tex
\section{Network Embeddings}
\label{sec:greedyEm}

Our solution builds upon previous work in the area of network embeddings, which assign coordinates to nodes
with the goal  of structuring networks, e.g., for efficient routing in wireless sensor networks or as an alternative to the current IP layer in content-centric networking.
We first introduce some notation, then explain the principal concepts of network embeddings, and conclude by describing specific algorithms. In particular, we detail the PIE embedding \cite{herzen2011scalable}, which we modify in Section \ref{sec:voute-design} to allow for anonymity.

\subsection{Basic Terminology}

In the remainder of the paper, we represent an overlay network by a graph $G=(V,E)$ with \emph{nodes} $V$ and \emph{links} or \emph{edges} $E \subset V \times V$.
Because we require mutual trust for connection establishment in F2F overlays, the network is bidirectional.
We denote the neighbors of $u$ by $N^G_u=\{v \in V: (u,v) \in E\}$.
Therefore, a \emph{Friend-to-friend (F2F) overlay} is an overlay such that the set of links is given by pairs of nodes sharing a real-world trust relationship.
Embedding algorithms heavily rely on \emph{spanning trees}, connected
subgraphs $ST=(V, E^T)$ of $G$ such that $|E^T|=|V|-1$.
In such a (spanning) tree, one node $r$ is designated as the \emph{root} and the position of nodes are described based on their relation to the root.
In particular, the \emph{level} or \emph{depth} of a node $u$ is given by the length of the path from $u$ to $r$.
If $u$ is not the root, the \emph{parent} of $u$ is defined to be a neighbor $v \in N^{ST}_u$ with a shorter path to $r$ than $u$,
whereas the remaining neighbors are $u$'s \emph{children}.
A node without children is called a \emph{leaf}, whereas nodes with children are called \emph{internal nodes}.

\subsection{Concept}

Now, we define the concept of network embeddings and in particular \emph{greedy} network embeddings.
In the following, let $G=(V,E)$ be a network and $(\X,\dist)$ be a metric space with a distance $\dist$.
A network embedding is defined as a function $\id: V \rightarrow \X$ assigning each node a coordinate.
The problem of enabling routing in a connectivity-restricted network has been addressed by the design of \emph{greedy embeddings}. Greedy embeddings \cite{papadimitriou2004conjecture} are coordinate assignments, such that for any source-destination pair $(s,t) \in V \times V$ with $s\neq t$, a neighbor $u$ of $s$ exists such that $\dist(u,t) < \dist(s,t)$. 
We say that $u$ is closer to $t$ than $s$ with regard to $\dist$.
As a consequence, straight-forward greedy routing  is guaranteed to find a route from $s$ to $t$.

Though there exists a multitude of greedy embedding algorithms, they all follow the same four abstract steps: i) Construct a spanning tree $T$, ii) Each internal node in $T$ enumerates its children, iii)  The root receives a predefined coordinate, iv)  Children derive their coordinate from the parent's coordinate and the enumeration index assigned by the parent (e.g. \cite{kleinberg2007geographic, cvetkovski2009hyperbolic, eppstein2009succinct, herzen2011scalable}).
The coordinates are then distributed such that the embedding of the spanning tree is greedy, as specified for the PIE embedding below.
Subsequent to the coordinate assignment, nodes consider all neighbors, including those that are neither parent nor child, for the routing.
So, routing is not restricted to tree edges. We call non-tree edges \emph{shortcuts} because they allow for a faster reduction of the distance and shorter routes than predicted by the distance in the tree. 

In the following, we consider the construction and stabilization costs for such greedy embeddings.
A spanning tree is constructed by i) selecting a root node using a distributed leader election protocol such as \cite{perlman1985algorithm,sirivianos2007non}, and ii) building the tree from the root.
In this manner, it is possible to construct a spanning tree with $\calO(n \log n)$ messages for a graph of diameter $\calO(\log n)$
\cite{perlman1985algorithm}, though integrating protections against nodes aiming to cheat the root selection protocol such as \cite{sirivianos2007non} require a higher cost.
Various embeddings \cite{cvetkovski2009hyperbolic, eppstein2009succinct, herzen2011scalable} are able to react to dynamics without computing the complete embedding whenever the topology changes.
 New nodes join the trees as leaves, requiring only a constant overhead for contacting one of their neighbors to be their parent and receiving a coordinate from said parent.
If any node but the root leaves, only its descendants have to reconnect.
We show that the stabilization overhead then scales linearly with the tree depth rather than linear with the number of participants.

\subsection{Existing Approaches}
Though embedding algorithms generally rely on a spanning tree and assign coordinates according to the tree structure, the nature of the assigned coordinates is highly diverse: Embeddings into hyperbolic space such as \cite{kleinberg2007geographic, cvetkovski2009hyperbolic, eppstein2009succinct} allow embedding in low-dimensional spaces.
However, proposed hyperbolic embeddings are extremely complex and do not scale with regard to the number of bits required for coordinate representation \cite{eppstein2009succinct}.
Custom-metric approaches have been designed to overcome these shortcomings.
The custom-metric embedding PIE \cite{herzen2011scalable} assigns an empty vector as the root coordinate.
Child coordinates are then derived from the parent coordinate by concatenating the parent coordinate with the index assigned to the child by the parent, potentially weighted with the cost of the parent-child edge if such weights are given.
In this manner, a node $s$'s coordinate represents the route from the root to $u$. 
Consequently, the distance $\dist$ is  given by the hop distance of two nodes in the tree.
An example for the PIE embedding in unweighted graphs is displayed on the left side of Figure \ref{fig:em}.
Whereas routing in greedy embeddings is highly efficient in comparison to non-greedy embeddings \cite{hofer2013greedy}, neither anonymity nor resilience has been considered in suitably manner.

%% file: design-trees.tex
\section{Design}
\label{sec:voute-design}

Our main contribution lies in proposing multiple greedy embeddings with anonymous return addresses and a virtual overlay on top of the embeddings.
In the following, we present our system, in particular
\begin{itemize}
\item a spanning tree construction and stabilization algorithm for multiple parallel embeddings, 
\item an embedding algorithm providing efficiency as well as allowing for improved censorship-resistance through a modified distance,    
\item an address generation algorithm $\ano$ enabling receiver anonymity, and
\item a virtual overlay design based on embeddings which allowing balanced content distribution and efficient content retrieval. 
\end{itemize}

\subsection{Tree Construction and Stabilization}
\label{sec:voute-trees}

In this section, we show how we construct and stabilize $\trees$ parallel spanning trees. In the next section, we then describe how to assign coordinates on the basis of these trees.
We want to increase the robustness and censorship-resistance by using multiple trees. In order to ensure that the trees indeed offer different routes, our algorithm encourages nodes to select different parents in each tree if possible.   
Our algorithm design follows similar principles as the provable optimally robust and resilient tree construction algorithm for P2P-based video streaming presented in \cite{brinkmeier2009optimally}.
However, the algorithm assumes that nodes can change their neighbors. Thus, we cannot directly apply the algorithm nor the results.
In the following, we first discuss the tree construction and then the stabilization.


\paragraph{Tree Construction:} We divide the construction of a tree into two phases i) selecting the root, and ii) building the tree starting from the root.
We can apply \cite{perlman1985algorithm} for the root election, which achieves a communication complexity of $\calO\left(n \log n\right)$.
Our own contribution lies in ii) the tree construction after the root node has been chosen. 

We now shortly describe the idea of our algorithm and then the actual algorithm.
A node $u$ that is not the root receives messages from its neighbors when they join a tree and are hence available as parent nodes. 
There are two questions to consider when designing an algorithm governing $u$'s reaction to such messages, called invitations in the following.
First, $u$ has to decide if and when it accepts an invitation. Second, $u$ has to select an invitation in the presence of multiple invitations.  

For the second question, $u$ always prefers invitation from nodes that have been their parent in less trees with the goal of constructing different trees and increasing the overall number of possible routes. Increasing the number of routes allows the use of alternative routes if the request can not be routed along the preferred route due to a failed or malicious node. 
If two neighbors are parents in the same number of trees, $u$ can either select one randomly or prefer the parent closer to the root.
Choosing a random parent reduces the impact of nodes close to the root but is likely to lead to longer routes and thus a lower efficiency.

Coming back to the first question of if and when $u$ accepts invitations, $u$ should always accept an invitation of a neighbor $v$ that is not yet a parent of $u$ in any tree in order choose different parents as often as possible. 
In contrast, if $v$ is already a parent, $u$ might wait for the invitation of a different neighbor.
However, it is unclear if it is possible for all neighbors of $u$ to ever become a parent. For example, a neighbor of degree $1$ is only
a parent if it is the root. 
In order to overcome this dilemma, $u$ periodically probabilistically decides if it should accept $v$'s invitation or wait for another invitation.
So, $u$ eventually accepts an invitation but does provide alternative parents the chance to send an invitation.

Now, we describe the exact steps of the algorithm. 
The algorithm is a round-based invitation protocol for the tree construction.
After a node $u$ is included in the $i$-th tree, $u$ sends invitations $(i,u)$ to all its neighbors inviting them to be its children in tree $i$.
When $u$ receives an information $(j,w)$ for the $j$-th tree from a neighbor $w$, it saves the invitation if it is not yet contained in tree $j$
and otherwise stores it. The invitation can still be used if $u$ has to modify its parent selection later. 
In each round, a node $u$ considers all invitations for trees it is not yet part of, as described in Algorithm \ref{algo:Tconst}.
Let $pc(v)$ be number of trees for which a neighboring node $v$ is a parent of $u$.
If $u$ has received invitations from neighbors $v$ with minimal $pc(v)$ among all neighbors, $u$ accepts one of those invitations (Lines \ref{algo:all1}-\ref{algo:all2}).
In the presence of multiple invitations, we experiment with two selection strategies: i) Choosing a random invitation, and ii) Choosing a random invitation from a node on the lowest level. 
The latter selection scheme requires that the invitations also detail the level of the potential parent node in the tree.
If $u$ does not have an invitation from any node with minimal $pc(v)$, $u$ nevertheless accepts an invitation with probability $q$
in order to guarantee the termination of the tree construction.
If $u$ accepts a parent, it selects a node $v$ that has offered an invitation and has the lowest $pc(v)$ among neighbors with outstanding invitations (Lines \ref{algo:limited1}-\ref{algo:limited2}).
In this manner, we guarantee the convergence of tree construction. 

The acceptance probability $q$ is essential for the diversity and the structure of the trees: For a high $q$, nodes quickly accept invitations leading to trees of a low depth and thus short routes. However, in the presence of an attacker acting as the root of all or most trees, the trees are probably close to identical, resulting in a low censorship-resistance.
A lower acceptance probability $q$ increases the diversity but entails longer routes.
Thus, a low $q$ results in a higher communication complexity and at some point decreases the robustness due to the increased likelihood of encountering failed nodes on a longer route. 
In Section \ref{sec:voute-perfbounds}, we show that the constructed trees are of a logarithmic depth such that  we indeed maintain a routing complexity of $\calO(\log n)$.
Note that Algorithm \ref{algo:Tconst} does not assume that all trees are constructed at the same time. Rather, individual trees can be (re-)constructed while the remaining trees impact the parent choice in the new tree but remain unchanged.

\begin{minipage}[ht]{.9\linewidth}
\begin{algorithm}[H]
\caption{\small constructTreeRound()}
\label{algo:Tconst}
\begin{algorithmic}[1]
{\small
\STATEx  \COMMENT{\textit{Internal state: Set $I$ of invitations, acceptance probability $q$, $pc:N_u \rightarrow \mathbb{N}_0$ number of times neighbor is parent}}
 \STATE $PP \leftarrow \{(i,w) \in I: \forall \mathbf{v \in N_u}: pc(w)\leq pc(v) \}$ \label{algo:all1} 
 \IF{$PP$ is not empty}
  \STATE Select invitation in $PP$ to answer \label{algo:all2}
 \ELSE
 \STATE $r \leftarrow $ uniform random number
 \IF{$r \leq q$}
 \STATE $PQ \leftarrow \{inv=(i,w) \in I: \forall \mathbf{(j,v) \in I}: pc(w)\leq pc(v) \}$ \label{algo:limited1} 
 \STATE Select invitation in $PQ$ to answer \label{algo:limited2}
 \ENDIF
 \ENDIF
} 
\end{algorithmic}
\end{algorithm}
\vspace{0.5em}
\end{minipage}

\paragraph{Stabilization:} Now, we consider the stabilization of the trees when nodes join and leave.
Stabilizing the trees efficiently, i.e., repairing them locally rather than reconstructing the complete tree whenever the topology changes, is essential for efficiency.
Joining nodes can be integrated in a straight-forward manner by connecting to their neighbors as children, again trying to maximize the diversity of the parents. 
For this purpose, nodes record the time, i.e., the round in our abstract time model, they joined the tree. Now, when a new node $u$ joins, it requests its neighbors' coordinates and these timestamps for all trees. 
Based on this information, $u$ can simulate Algorithm \ref{algo:Tconst} locally, ensuring that its expected depth in the tree is unaffected by its delayed join.
When a node departs, all its children have to choose a different parent and inform their descendants of the change.
In order to prevent a complete subtree from being relocated at an increased depth, the descendants may also select a different parent.
The selection of the new parent again follows Algorithm \ref{algo:Tconst} but only locally re-establishes the trees affected by the node departure. 
We show that the stabilization complexity considering any node but the root is linear in the terms of average depth of the node in the trees in Section \ref{sec:voute-perfbounds}. 

We formally prove that the above stabilization algorithm indeed only introduces only logarithmic complexity in Section \ref{sec:voute-perfbounds}. We now present the embedding algorithm, which in agreement with the presented tree construction algorithm, assigns coordinates within subtrees independently of the remaining subtrees to allow for local stabilization.

%% file: design-embedding.tex
\subsection{Embedding and Routing}
\label{sec:voute-embedding}

 \begin{figure}
 \centering
 \includegraphics[width=0.7\linewidth]{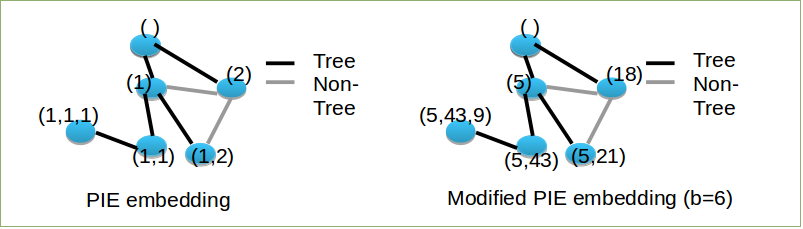}
 \caption[Original PIE vs. modified PIE]{Original PIE and modified PIE coordinates using $b=6$-bit number}
 \label{fig:em}
 \end{figure}

In this section, we show how we assign coordinates in a spanning tree and how to route based on these coordinates.
As we want to prevent an attacker from guessing the coordinate of a receiver, we require a certain degree of in-determinism in the coordinate assignment.
We thus choose a slightly modified version of the unweighted PIE embedding \cite{herzen2011scalable}, which we have introduced in Section \ref{sec:greedyEm}.
Our main modification lies the use of in-deterministic coordinates in order to prevent an adversary from guessing the coordinate and thus undermining the anonymization schemes presented in the next section.
The routing algorithm corresponds to the greedy routing with backtracking. In addition to the tree distance in  \cite{herzen2011scalable}, we also present a second distance preferring nodes with a long common prefix and thus avoiding routes via nodes close to the root whenever possible.
In this manner, we increase robustness and censorship-resistance because the routing algorithm considers alternative routes and the impact of strategically choosing a position close to the root is reduced.
In the following, we subsequently present the embedding algorithm, the distance functions, and the routing with backtracking.

\paragraph{Embedding Algorithm: } Embeddings are performed on each of the $\trees$ trees independently, so that we only consider one embedding $\id$.
Throughout this section, let $b$ be a sufficiently large integer, $PRNG$ a pseudo-random number generator with values in $\mathbb{Z}_2^b$, and $h: \{0,1\}^* \rightarrow H$ a cryptographically secure hash function.
We describe the embedding algorithm, then the distance used for routing, and last, the backtracking procedure, which allows highly resilient routing despite failures.

We now describe the embedding algorithm for one tree. 
The coordinate assignment starts at the root and then spreads successively throughout the tree. 
After a spanning tree has been established, the root $r$ is assigned an empty vector as a coordinate $\id(r)=()$. 
In the next step, each child $v$ of the root generates a random $b$-bit number $a \in \mathbb{Z}^b_2$ such that its coordinate is $\id(v)=(a)$.
Here, our algorithm differs from the PIE embedding because it uses random rather than consecutive numbers, thus preventing an adversary from guessing the coordinate in an efficient manner. 
Subsequently,  nodes in the tree are assigned coordinates by concatenating their parent's coordinate with a random number.
So, upon receiving its parent coordinate $\id(p(v))=(a_1, \ldots, a_{l-1})$, a node $v$ on level $l$ of the tree  obtains its coordinate $\id(v)=(a_1, \ldots, a_{l-1}, a_l)$ by adding a random $b$-bit number $a_l$.
The coordinate space is hence given by all vectors consisting of $b$-bit numbers, i.e., $\X=\{(a_1, \ldots, a_{l-1}, a_l): l \in \mathbb{N}_0, a_i \in \{0,1\}^b\}$.
Figure \ref{fig:em} displays the differences between the original PIE embedding and our variation.

Note that the independent random choice of the  $b$-bit number $a \in \mathbb{Z}^b_2$ might lead to two nodes having the same coordinate.
Thus, $b$ should be chosen such that the chance of equal coordinates should be negligible.
If two children nevertheless select the same coordinate, the parent node should inform one of them to adapt its choice.
Note that allowing the parent to influence the coordinate selection in this manner does not really increase the vulnerability to attacks,
as the parent can achieve at least the same damage by constantly changing its coordinate. Such constant changes can be detected easily, so that nodes should stop selecting such nodes as parents. 
In general, by moving the choice of the last coordinate element from the parent to the child, we automatically reduce the impact of a malicious parent as it can not determine the complete coordinate of the child.

\paragraph{Distances: }
We still need to define distances between coordinates in order to apply greedy routing. 
For this purpose, we consider two distances on $\X$.
Both rely on the common prefix length $cpl(x_1, x_2)$ of two vectors $x_1$ and $x_2$ and the coordinate length $|x_1|$.

First, we consider the tree distance $\delta_{TD}$ from \cite{herzen2011scalable}, which gives the length of  path between the two nodes in the tree, i.e., 
\begin{align}
\label{eq:TD}
\delta_{TD}(x_1,x_2) = |x_1|+|x_2|-2 cpl(x_1, x_2).
\end{align}

Secondly, the common prefix length can be used as the determining factor in the distance function, i.e., for  a constant $\lengthAd$ exceeding the length of all node coordinates in the overlay, we define 
\begin{align}
\label{eq:CPLD}
\delta_{CPL}(x_1,x_2) =
\begin{cases} 
\lengthAd-cpl(x_1,x_2)-\frac{1}{|x_1|+|x_2|+1}, & x_1 \neq x_2 \\
0, & x_1 = x_2 
\end{cases}.
\end{align}
The reason for using the common prefix length rather than the actual tree distance is the latter's preference of routes passing nodes close to the root in the tree. 
In this manner, these nodes on these routes are very influential, so that adversaries can gain a large impact from gaining such a position.
In contrast, $\delta_{CPL}$ prefers possibly longer routes by always forwarding to a node within the same subtree as the destination and avoids central nodes in the tree. An example of the difference between the two distances and the impact on the discovered routes is displayed in Figure \ref{fig:dist}.

\begin{figure}
\centering
\includegraphics[width=0.7\linewidth]{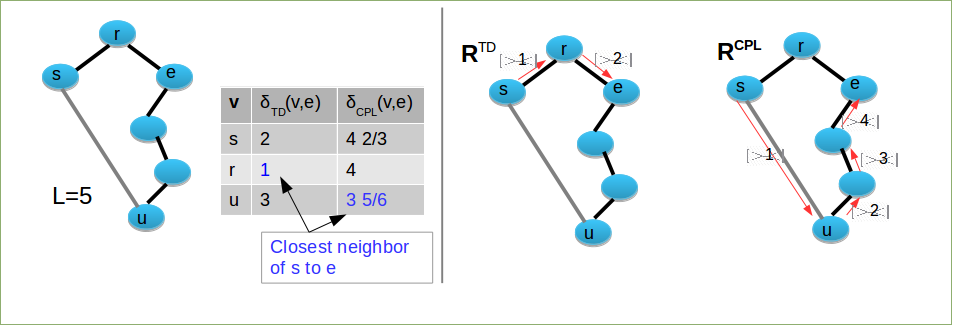}
\caption[Distances in Tree-based Embeddings]{Tree distance (TD) $\delta_{TD}$ and common prefix length based distance $\delta_{CPL}$ when routing from node $s$ to $\e$: $\delta_{CPL}$ prefers nodes in the same subtree as the destination, leading to better censorship-resistance at the price of longer routes. The table gives the distances in the first hop for $s$ and its neighbors $r$ and $u$.}
\label{fig:dist}
\end{figure}

\paragraph{Greedy Routing in Multiple Embeddings: }
We route in $1 \leq \tau \leq \trees$ trees in parallel. 
More precisely,  given a vector of coordinates $(\id_1(\e), \ldots , \id_\trees(\e))$, the sender $s$ selects $\tau$ coordinates and sends a request for each of them. $\s$ can either select $\tau$ embeddings uniformly at random or choose the embeddings so that 
the distance of the neighbor $v_i$ with the closest coordinate to $\id_i(\e)$ is minimal. 
The latter choice might result in shorter routes due to the low distance in the embedding. 

The routing processes in each embedding independently.
Nodes forward the request to the neighbor with the closest coordinate in the respective embedding.
Thus, in order for the nodes on the route to forward the request correctly, 
the request has to contain both the coordinate $\id_i(v)$ and the index $i$ of the embedding.
In practice, we can achieve a performance gain by including multiple coordinates and embedding indices in one message if the next hop in two or more  embeddings are identical.
For now, we assume that one message is sent for each embedding for simplicity.

We optionally increase the robustness and censorship-resistance of the routing algorithm by allowing backtracking if the routing gets stuck in a local minimum of the distance function due to failures or intentional refusal to forward a request. 
For this purpose, all nodes remember their predecessor  on the routing path as well as the neighbors they have forwarded the request to. 
If all neighbors closer to the target have been considered and have been unable to deliver the request, the node reroutes the request to its predecessor for finding an alternative path.
The routing is thus only considered to be failed if the request returns to its source $\s$ and cannot be forwarded to any other neighbor.
In this manner, all \emph{greedy paths}, i.e., all paths with a monotonously decreasing distance to the target, are found.

\begin{minipage}[ht]{.9\linewidth}
\begin{algorithm}[H]
\caption{\small route()}
\label{algo:route}
\begin{algorithmic}[1]
{\small
\STATEx  \COMMENT{\textit{Input: current node $u$, message $msg$ from node $w$, tree index $i$, target coordinate $x_i$}}
\STATEx  \COMMENT{\textit{Internal state: set $S(msg)$ of nodes $u$ forwarded $mess$ to, predecessor $pred(msg)$, distance $\delta$}}
\IF{$id_i(u) == x_i$}
\STATE  Routing succeeds \label{algo:success}
\ELSE
\commIt{Store predecessor unless backtracking}
\IF{not $S(msg)$ contains $w$}
\STATE $pred(msg) \leftarrow w$ \label{algo:pred}
\ENDIF 
\commIt{Determine closest neighbors}
\STATE $C \leftarrow argmin_{v \in N_u\setminus S(msg)} \delta(id_i(v), x_i)$ \label{algo:g1}
\STATE $next \leftarrow $ random element in C \label{algo:g2}
\IF{$\delta(id_i(v), x_i) > \delta(id_i(next), x_i)$} 
\STATE Forward $msg$ to $next$ \commItLine{Forward if improvement} \label{algo:forward}
\ELSE
\IF{$pred(msg)$ is set}
\STATE  Forward $msg$ to $pred(msg)$ \commItLine{Backtrack}\label{algo:backtrack}
\ELSE
\STATE Routing failed \label{algo:fail}
\ENDIF
\ENDIF
\ENDIF 
} 
\end{algorithmic}
\end{algorithm}
\vspace{0.5em}
\end{minipage}

Algorithm \ref{algo:route} gives the pseudo code  describing one step of  the routing algorithm, including the backtracking procedure. 
When receiving a message $msg$, the node $u$ first checks if it is the receiver of $msg$, thus successfully terminating the routing (Line \ref{algo:success}).
If $u$ is not the receiver, it determines if the routing is currently in the backtracking phase by checking if  $u$ has previously forwarded $msg$ to the sender $w$. 
Otherwise, it stores the sender of $msg$ as a predecessor for potential later backtracking (Line \ref{algo:pred}).
In the manner of greedy routing, $u$ selects the closest neighbor to the target coordinate. In the presence of several closest neighbors, 
$u$ picks one of them uniformly at random (Lines \ref{algo:g1}-\ref{algo:g2}).
Note that in the presence of failures, the embedding can lose its greediness.
Hence, to avoid loops, $u$ only forwards the request to that neighbor if it is indeed closer (Line \ref{algo:forward}).
Otherwise, $u$ contacts its predecessor (Line \ref{algo:backtrack}) or forfeits the routing if no such predecessor exists (Line \ref{algo:fail}), i.e., if $u$ is the source of the request.

This completes the description of the routing and stabilization functionalities. However, up to now, we used identifying coordinates rather than anonymous addresses. 

%% file: design-returnAdd.tex
\subsection{Anonymous Return Addresses}
\label{sec:voute-returnAdd}

In this section, we introduce our address generation algorithm  for generating anonymous return addresses but do not reveal the receiver of the request. 
For this reason, we call the generated addresses \emph{route preserving (RP) return addresses}. 
Based on these return addresses, we specify two routing algorithms $\routeTD$ and $\routeRAPCPL$ for routing a request containing a return address. 
The return addresses allow a node to determine the common prefix length of their neighbor's coordinates and the receiver coordinate, which allows the node to determine the closest neighbor. Hence, $\routeTD$ and $\routeRAPCPL$ correspond to Algorithm \ref{algo:route} for the two distance function $\delta_{TD}$ and $\delta_{CPL}$ when using return addresses rather than a receiver coordinates. 
After describing the algorithm, we show that the return addresses indeed preserve routes. 


\paragraph{Return Address Generation:} Return addresses are generated in three steps:
\begin{enumerate}
\item Padding the coordinate 
\item Applying a hash cascade to obtain the return address 
\item Adding a MAC
\end{enumerate}
Algorithm \ref{algo:rap} displays the pseudo code of the above steps.

\begin{minipage}[ht]{.9\linewidth}
\begin{algorithm}[H]
\caption{\small generateRP()}
\label{algo:rap}
\begin{algorithmic}[1]
{\small
\commIt{Input: coordinate $x=(a_1, \ldots , a_l)$, seed $s,s_{pad}$}
\commIt{Internal State: key $\keymac(v)$, $h$, $PRNG$}
\STATE $\tilde{k} \leftarrow PRNG(s)$ 
\STATE $d_1 \leftarrow h(\tilde{k} \xor a_1)$
\FOR{$j=2\ldots \lengthAd$}
\IF{$j \leq l$}
\STATE $a'_j \leftarrow a_j$
\ELSE
\STATE $a'_j \leftarrow PNRG(s_{pad}+ j)$ \commItLine{Padding}
\ENDIF
\STATE $d_j \leftarrow h(d_{j-1} \xor a'_j)$ \commItLine{Hash cascade}
\ENDFOR
}
\STATE \small{$mac \leftarrow h(\keymac(v)||d_1||d_2 ||\ldots || d_\lengthAd)$ \commItLine{MAC}}
\STATE \small{Publish $y=(d_1, \ldots , d_\lengthAd), \tilde{k}, mac$} 
\end{algorithmic}
\end{algorithm}
\vspace{0.5em}
\end{minipage}

The first step of the return address generation prevents an adversary from identifying coordinates based on their length.
A node  $v$ pads its coordinate $x=(a_1,\ldots , a_l)$ by adding random elements $a'_{l+1}, \ldots ,a'_\lengthAd$. 
More precisely, $v$ selects a seed $s_{pad}$ for the pseudo-random number generator $PRNG$ and obtains the padded coordinate 
$x'=(a'_1, \ldots , a'_l, a'_{l+1}, \ldots, a'_\lengthAd)$ with
\begin{align*}
a'_j = \begin{cases}
a_j,& j \leq l \\
PRNG(s_{pad} \xor j), & j > l
\end{cases}.
\end{align*}
In order to ensure that the closest node coordinate to $x'$ is indeed $x$, $v$ recomputes the padding with a different seed
if $a'_{l+1}$ is equal to the $l+1$-th element of a child's coordinate \footnote{We exclude this step in Algorithm \ref{algo:rap} for  increased readability}.
Afterwards, $v$ chooses a different seed $s$ for the construction of the actual return address and generates 
$\tilde{k}= PRNG(s) \in \keysPart = \mathbb{Z}^b_2$.
$v$ then executes the local function $hc: \X \rightarrow \Y=H^\lengthAd$ in order to obtain a vector $y$ with elements in $H$.
The $i$-th element of $y=(d_1, \ldots , d_\lengthAd)$ is given by 
 \begin{align}
 \label{eq:cash}
 d_j =
\begin{cases}
h(\tilde{k} \xor a'_1), & j=1 \\
h(d_{j-1} \xor a'_j), & j =2 \ldots \lengthAd
\end{cases}.
\end{align}
We call the pair $(y, \tilde{k})$ a \emph{return address}, which can be used to find a route to the node with coordinate $x$.
Before publishing the return address, $v$ adds a MAC  
$mac(y_i, \keymac(v))=$ $h(d_1 || \ldots d_\lengthAd || \keymac(v))$ for a private key $\keymac(v)$ to prevent malicious nodes from faking return addresses and gaining information from potential replies.
Last, $v$ publishes the return address $(y, \tilde{k})$ and the MAC.


\paragraph{Routing Algorithms:} Now, we determine diversity measures $\distIndex{RP-TD}:\X \times \Y \rightarrow \mathbb{R}_{+}$ and $\distIndex{RP-CPL}:\X \times \Y \rightarrow \mathbb{R}_{+}$ in order to compare coordinates $x$ and $y$ with regard to $\delta_{TD}$ and  $\delta_{CPL}$.
The diversity measure then assumes the role of the distance $\delta$ in Algorithm \ref{algo:route}.
\footnote{Note that a diversity measure is not a distance because it i) is defined for two potentially distinct sets $\X$ and $\Y$, and ii) is not symmetric.}

In order to define a sensible diversity measure, note that for any coordinate $c$ and return address $y$ for coordinate $x$, we
have  $cpl(x,c)=cpl(y, hc(c,\tilde{k}))$.
We thus can define the diversity measure in terms of the common prefix length in the same manner as the distance. 
More precisely, for $* \in \{TD, CPL\}$, the diversity $\distIndex{RP-*}(y, \tilde{k}, c)$ for of a coordinate $c$ to the return address $y$ is
\begin{align}
\label{eq:distRap}
\distIndex{RP-*} (y, \tilde{k}, c) = \distIndex{*}(y_i, hc(c,\tilde{k})).
\end{align}
In practice, $u$ can increase the efficiency of the computation by only determining $hc(c,\tilde{k})$ up to the first element in which it disagrees with $y$.
Thus, we now have two possible realizations of the routing algorithm $\routeN$, namely $\routeTD$ and $\routeRAPCPL$.
Given the RP return address $(y, \tilde{k})$ of the destination $\e$, $\routeTD$ and $\routeRAPCPL$ forward the message to the neighbor $v$  with the lowest diversity measure $\distIndex{RP-TD} (y, \tilde{k}, \id(v))$ and $\distIndex{RP-CPL}(y, \tilde{k}, \id(v))$, respectively.

\paragraph{Proving Route Preservation:} We now prove formally that the above return addresses preserve routes. For this purpose, we first define the notion of preserving a property of a coordinates. 
Note that we
 \begin{definition}
\label{def:preserving}
Let $\Q_u:\mathcal{P}(\X) \times \X \rightarrow \mathcal{P}(\X)$ be a local function of node $u$ in a graph $G=(V,E)$
Given a set $C \subset \X_V= \{v \in V: \id(v)\}$ of node coordinates  and a target coordinate $x \in \X$, $\Q_u$ returns a subset $C' \subset C$.
A return address $(y, \tilde{k})$ for a coordinate $x$ is said to \emph{preserve $\Q$} if for all $u \in V$, there exists a function $\Q'_u: \mathcal{P}(\X) \times \Y \times \keysPart \rightarrow \mathcal{P}(\X)$ such that for all $C \subset \X$
 \begin{align*}
 \Q'_u(C,y, \tilde{k}) = \Q_u(C,x).
 \end{align*}
\end{definition}
The notion of \emph{route preserving (RP)} return addresses now follows if we choose the function $Q_u$ to return the neighbors with the closest coordinates to $\pre(y, \tilde{k})$.
\begin{definition}
\label{def:rap}
Let 
\begin{align}
\label{eq:ra}
\begin{split}
&ra_u \colon \mathcal{P}(\X) \times \X \rightarrow \mathcal{P}(\X), \\
&ra_u(C,x) = argmin_{c \in C}\{\delta(c,x)\} 
\end{split}
\end{align}
determine the closest coordinates in a set $C$ to a coordinate $x$.
A return address $(y, \tilde{k})$ is called route preserving (RP) (with regard to $\delta$) if it preserves $ra$. 
\end{definition} 

Based Definition \ref{def:rap}, we can now show that Algorithm \ref{algo:rap} generates RP return addresses.
\begin{theorem}
\label{thm:rpRA}
Algorithm \ref{algo:rap} generates RP return addresses with regard to the distances $\delta_{TD}$ and $\delta_{CPL}$.
\end{theorem}
\begin{proof}
In order to show that $(y, \tilde{k})$ preserves routes, we derive the relation between the diversity measures $\distIndex{RP-TD}$ and $\distIndex{RP-CPL}$, defined in Eq. \ref{eq:distRap}, and the corresponding distances $\delta_{TD}$ and $\delta_{CPL}$, defined in Eq. \ref{eq:TD} and Eq. \ref{eq:CPLD}, respectively.

Let $\pre(y, \tilde{k})$ denote the padded coordinate used to generate $y$, and let $x$ be the coordinate without padding.
In the following, we relate the distance of $x$ and a coordinate $c$ to the diversity measure of $(y, \tilde{k})$ and $c$.
We can assume that $cpl(\pre(y, \tilde{k}),c)=cpl(x,c) \leq |x|$, i.e., the common prefix length of the padded coordinate and $c$ is at most equal to the length of the original coordinate $x$.
A node with coordinate $c$ with $cpl(\pre(y, \tilde{k}),c) > |x|$ cannot exist in a valid embedding. 
More precisely, our embeddings algorithm ensures that coordinates are unique and a node $v$ ensures that the first element of the padding does not corresponds to the $|\id(v)|+1$-th element of a descendant's coordinate.  Thus, the coordinate $x$ is the unique closest coordinate of a node to the padded coordinate. 
Thus, we can indeed limit our evaluation to coordinates $c$ with  $cpl(\pre(y, \tilde{k}),c) \leq |x|$.

We start by considering the tree distance $\delta_{TD}$. By Eq. \ref{eq:distRap}, we have
\begin{align*}
\distIndex{RP-TD}(y, \tilde{k},c) &= \lengthAd + |c| -2 cpl(\pre(y, \tilde{k}),c) \\
&= |x| + |c|-2 cpl(x,c) + (\lengthAd-|x|) \\
&=\delta_{TD}(x,c) + (\lengthAd-|x|).
\end{align*}
Hence, diversity measure and distance only differ by a constant independent of $c$.
Thus, any forwarding node can determine the closest coordinates to the destination in its neighborhood and thus
Algorithm \ref{algo:rap} generates RP return addresses with regard to $\delta_{TD}$.
 
For the distance $\delta_{CPL}$, we consider two coordinates $c_1$ and $c_2$ with $cpl(\pre(y, \tilde{k}),c_i)=cpl(x,c_i)$ for $i=1,2$. We show that i) $\delta_{CPL}(x,c_1)=\delta_{CPL}(x,c_2)$
iff $\delta_{RP-CPL}(y, \tilde{k},c_1)=\delta_{RP-CPL}(y, \tilde{k},c_2)$ and ii) $\delta_{CPL}(x,c_1)<\delta_{CPL}(x,c_2)$
iff $\delta_{RP-CPL}(y, \tilde{k},c_1)<\delta_{RP-CPL}(y, \tilde{k},c_2)$ . 
Thus, the return address $(y, \tilde{k})$ is RP as the comparison of two coordinates yields the same order when using the return address as for the original coordinate. 
For i) note that by Eq. \ref{eq:CPLD} $\delta_{CPL}(x,c_1)=\delta_{CPL}(x,c_2)$ implies that $cpl(x,c_1)=cpl(x,c_2)$ and $|c_1|=|c_2|$. Because
$cpl(\pre(y, \tilde{k}),c_i)=cpl(x,c_i)$, we have $\delta_{RP-CPL}(y, \tilde{k},c_1)=\delta_{RP-CPL}(y, \tilde{k},c_2)$.
The converse holds analogously by Eq. \ref{eq:distRap}. 
If ii) $\delta_{CPL}(x,c_1)<\delta_{CPL}(x,c_2)$, then Eq. \ref{eq:CPLD} implies that either $cpl(x,c_1)> cpl(x,c_2)$ or $cpl(x,c_1)= cpl(x,c_2)$ and $|c_1|<|c_2|$.
In the first case, the claim follows as $cpl(\pre(y, \tilde{k}),c_i)=cpl(x,c_i)$ and $\delta_{CPL}$ and $\delta_{RP-CPL}$ both prefer coordinates with a longer common prefix length.
For the second case, the claim follows under the assumptions $cpl(x,c_1)= cpl(x,c_2)$ and $cpl(\pre(y, \tilde{k}),c_i)=cpl(x,c_i)$, because
\begin{align*}
 & \delta_{CPL}(x,c_1)<\delta_{CPL}(x,c_2)\\
 \iff & \lengthAd - cpl(x,c_1) -\frac{1}{|x|+|c_1|+1} < \lengthAd - cpl(x,c_2)-\frac{1}{|x|+|c_2|+1} \\
\iff &-\frac{1}{|x|+|c_1|+1} < -\frac{1}{|x|+|c_2|+1} \\
\iff &  |x| + |c_1|+1 < |x| + |c_2| +1 \\
\iff & \lengthAd + |c_1|+1 < \lengthAd + |c_2| +1  \\
\iff &-\frac{1}{\lengthAd+|c_1|+1} < -\frac{1}{\lengthAd+|c_2|+1} \\
\iff &\delta_{CPL}(\pre(y, \tilde{k}),c_1)<\delta_{CPL}(\pre(y, \tilde{k}),c_2) \\
\iff & \delta_{RP-CPL}(y, \tilde{k},c_1)<\delta_{RP-CPL}(y, \tilde{k},c_2)
\end{align*}
Hence for both cases i) and ii), Algorithm \ref{algo:rap} generates RP return addresses with regard to $\delta_{CPL}$.
\end{proof}

Up to now, we have only considered route preserving return addresses generated by padding coordinates and applying a hash cascade. 
Optionally, an additionally layer of symmetric encryption can be added, preventing a node $v$ from deriving the actual length of the common prefix. Rather, $v$ can only determine if a neighbor is closer to the destination than $v$ itself.
However, we show the same degree of anonymity for for both algorithms, so that the additional layer does not result in a provably higher level of anonymity. 
Furthermore, the additional layer reduces the efficiency as nodes select one closer neighbor at random rather than the closest neighbor.
For this reason, the advantage of the additional layer is limited, so that we focus on RP return addresses here and defer the further obfuscation of coordinates to the appendix. 

We prove that Algorithm \ref{algo:rap} indeed enables receiver anonymity in Section \ref{sec:voute-eval-ano}.

%% file: design-content.tex
\subsection{Content Storage}
\label{sec:voute-content}
In order to store content, we use a distributed hash table (DHT).
As nodes can not communicate directly, they store tree addresses in their routing tables and leverage the tree routing.
In this manner, we do not require maintenance-intensive tunnels like \cite{mittal2012x} and  \cite{vasserman2009membership}.
Note that we only sketch the solution for content storage and retrieval because our focus lies in improving the quality of the greedy embeddings for messaging between nodes. 
In the following, we first present the idea of our design and then a realization based upon a recursive Kademlia.

\paragraph{General Design: }
Nodes establish a DHT by maintaining a routing table of (virtual) overlay connections. The routing table contains entries
correspond to a DHT coordinate and corresponding return addresses. 
Nodes communicate with their virtual neighbors by sending requests in any of the $\trees$ embedding. 

New routing table entries are added by routing for a suitable virtual overlay key, as done in \cite{mittal2012x} for the tunnel discovery. 
However, after the routing terminates, the discovered nodes send back their return addresses rather than taking the routing path as a new tunnel.
In this manner, the length of routes between virtual overlay neighbors only depends on the trees and does not increase over time.
The exact nature of the neighbor discovery, the routing algorithm $\find$, and the stabilization of the virtual overlay depend on the  specifications of the DHT. 

\paragraph{Kademlia: }
In our evaluation in Sections \ref{sec:voute-eval-effi} and \ref{sec:voute-eval-resi}, we utilize a highly resilient recursive Kademlia \cite{heep2010r}.
In Kademlia, a node selects a \emph{Kademlia identifier} $ID(v)$ uniformly at random in the form of a 160-bit number.  
The distance between identifiers is equal to their XOR.
Nodes maintain many redundant (virtual) overlay connections to increase the resilience. 
More precisely, each node $v$ keeps a \emph{routing table} of $k$-buckets. The $j$-th bucket contains up to $k$ addresses of nodes $u$ so that the common prefix length of $ID(v)$ and $ID(u)$ is $j$.
Maintaining more than $1$ neighbor per common prefix length increases the robustness to failures and possibly even to attacks due to the existence of alternative connections. 

Based on such routing tables, efficient and robust content discovery is possible.  
Files are indexed by keys corresponding to the hash of their content, i.e., the algorithm $\anoC$ for the generation of file addresses is a hash function. A node $u$ requesting a file with key $f$ looks up the closest nodes  $v_1, \ldots, v_\alpha$ to $f$ in its routing table in terms of virtual overlay coordinates. 
Then, $u$ routes for each $v_i$ in the $\tau$ trees. Upon receiving one of the messages, $v_i$ returns $f$ via the same route if in possession of $f$.
Otherwise, $v$ forwards the message to the overlay neighbor closest to $f$, again using tree routing, and returns an acknowledgement message to $u$.
If a node on the route has already received the message via a parallel query, it returns a backtrack message such that the predecessor can contact a different node. 
Similarly, if a node does not receive an acknowledgement from its overlay neighbor in time, it selects an alternative node from its routing table if virtual neighbors closer to $f$ than $u$ exist.

Similarly, stabilization is realized in the same reactive manner as in the original Kademlia. Whenever a node successfully sends a message to an overlay neighbor, this neighbor returns an acknowledgement containing updated return addresses if any coordinates were changed.
If a node in the routing table cannot be contacted, the node removes the neighbor from the routing table. 
Depending on the implementation, it initializes a new neighbor discovery request for the prefix.
In addition, suitable neighbors encountered during routing are added to the routing table.

\paragraph{}
We have now presented the essential components of our design. In the following, we evaluate our design with regard to our requirements. The different layers of our system are displayed in Figure \ref{fig:layers}

  \begin{figure}
  \centering
  \includegraphics[height=7cm]{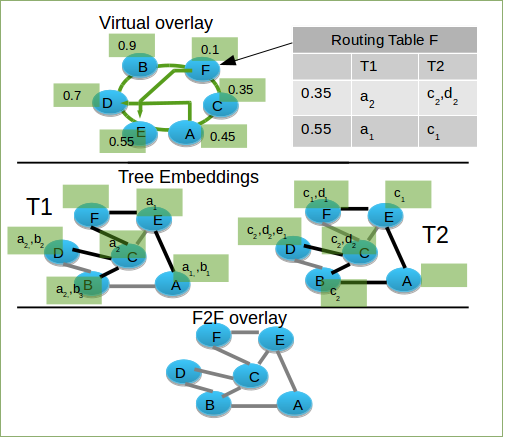}
  \caption[Layers of VOUTE]{Layers of VOUTE: 1) F2F overlay as restricted topology, 2) Tree embeddings $T1$ and $T2$ offer addressing for messaging, 3) Virtual overlays with tree addresses offer content sharing DHT routing based on tree addresses}
  \label{fig:layers}
  \end{figure}

%% file: eval-efficiency.tex
\section{Efficiency and Scalability}
\label{sec:voute-eval-effi}

In this section, we analyze the efficiency of our scheme with regard to routing complexity, stabilization complexity, and their evolution over time.

\subsection{Theoretical Analysis}
\label{sec:voute-perfbounds}

In the first part of this section, we obtain upper bounds on the expected routing length of the routing algorithms $\routeTD$ and $\routeRAPCPL$. The desired upper bound on the routing complexity follows by multiplying this bound for routing in one tree with $\tau$, the number of trees  used for parallel routing.
Afterwards, we consider the stabilization complexity $CS^\stab$ of the stabilization algorithm $\stab$ consisting of i) the local reconstruction of the trees according to Algorithm \ref{algo:Tconst} and ii) the assignment of new coordinates for the nodes affected by a change topology using the modified PIE embedding. 

\paragraph{Routing:} We consider both messaging between nodes as well as content discovery in the DHT. 

\begin{theorem}
\label{thm:PIE-routing}
Let $\id$ be a modified PIE embedding on a spanning tree of $G$ generated by Algorithm \ref{algo:Tconst} with parameters $\trees$ and $q$.
Furthermore, assume that the diameter of $G$ is $diam(G)=\calO(\log n)$.
The expected routing length  of Algorithm \ref{algo:route} is at most 
\begin{align}
\label{eq:R-TD}
\E(R^{TD}) = \calO\left(\frac{\trees}{q}\log n\right)
\end{align}
for the routing algorithm $\routeTD$, and
\begin{align}
\label{eq:R-CPL}
\E(R^{CPL}) = \calO\left(\left(\frac{\trees}{q}\right)^2 \log n\right)
\end{align}
 for $\routeRAPCPL$.
\end{theorem}

For the proof, we first show  Lemma \ref{lem:tconst}, which bounds the expected level of a node in trees constructed
by  Algorithm \ref{algo:Tconst}.
More precisely, we prove that the expected level of a node in any tree constructed by Algorithm \ref{algo:Tconst} is increased by at most  a constant factor in comparison to a breath-first-search.

\begin{lemma}
\label{lem:tconst}
Let $T$ be any of the $\trees$ trees constructed by Algorithm \ref{algo:Tconst} and $r$ the root of $T$.
Furthermore, denote by $sp_r(v)$ the length of the shortest path from $v$ to $r$, and let $L_T(v)$ be the level of $v$ in $T$. Then the expected value of $L_T(v)$ is bound by
\begin{align}
\label{eq:level}
\E(L_T(v)) \leq sp_r(v)\cdot \left(1+\frac{\trees}{q}\right).
\end{align}
\end{lemma}
\begin{proof}
We first give an upper bound on the expected number of rounds until a node $v$ accepts an invitation for $T$ after receiving the first invitation. Afterwards, we show Eq. \ref{eq:level} by induction.

In the first step, we denote the number of rounds until acceptance by $Y$.
In order to derive an upper bound on $\E(Y)$, we assume that $v$ does not receive any invitation that it can immediately accept, i.e., an invitation from neighbors $u$ with minimal parent count $pc(u)$. Thus, $v$ accepts one invitation with probability $q$ in each round. In the worst case, the $\trees$-th accepted invitation is for tree $T$.
The number of rounds thus corresponds to the sum of $\gamma$ identically distributed geometrically distributed random variables $X_1, \ldots , X_\trees$.
Here, $X_i$ is the number of trials until the first success of a sequence of Bernoulli experiments with success probability $q$, i.e., the number of rounds until an invitation is accepted. 
The random variable $X=X_1 + \ldots + X_\trees$ describes the number of trials until the $\trees$-th success and presents an upper bound on the expected number of rounds until acceptance of an invitation for tree $T$. 
We hence derive an upper bound on $\E(Y)$ by
\begin{align} 
\label{eq:EY}
\E(Y) \leq E(X) = \sum_{i=1}^\trees \E(X_i) = \trees \E(X_1) = \frac{\trees}{q}.
\end{align}

In the second step, we apply induction  on $l=sp_r(v)$. Note that the level of a node in the tree is at most the number of rounds until an invitation is accepted from the start of the protocol.
For $l=1$, the node $v$ receives an invitation from $r$ at round $1$ of the protocol because $v$ is a neighbor of the root node. 
In expectation, $v$ joins $T$ at round at most $1+\E(Y) \leq 1 +\frac{\trees}{q}$, which shows the claim for $l=1$.
Now, we assume Eq. \ref{eq:level} holds for $l-1$ and show that then it also holds for $l$.
The number of rounds $Z$ until the node $v$ at level $l$ accepts an invitation in tree $T$ is the sum of $Z_1$, the number of rounds until the first invitation is received, and $Z_2$ the number of rounds $v$ accepts after receiving the first invitation.
$v$ is the neighbor of a node $w$ with $sp_r(w)=l-1$ and receives an invitation from $w$ one round after $w$ joined $T$. 
So, $Z_1$ is bound by our induction hypothesis, and $Z_2$ is equal to $Y$ and hence bound by Eq. \ref{eq:EY}.
As a result,
\begin{align*}
\E(Z) = \E(Z_1)+1 + \E(Z_2) \leq  (l-1)\cdot \left( \frac{\trees}{q} + 1\right) + 1 + \frac{\trees}{q} + 1 = l \cdot \left( \frac{\trees}{q} + 1\right),
\end{align*}
and hence indeed Eq. \ref{eq:level} holds.
\end{proof}

Based on Lemma \ref{lem:tconst}, we now prove Theorem \ref{thm:PIE-routing}. The idea of the proof is to bound the routing length by a multiple of expected level of a node.

\begin{proof}
We consider the diversity measure $\delta_{RP-TD}$ first and then 
$\delta_{RP-CPL}$.

For $\delta_{RP-TD}$, the claim follows directly from Lemma \ref{lem:tconst} and Theorem 4.3 in \cite{herzen2011scalable}.
More precisely, the expected level of a node is at most $\calO\left( \frac{\trees}{q}\log n\right)$ assuming a diameter and hence maximal distance to the root of $\calO(\log n)$.
Recall that the distance $\delta_{TD}(\id(\s), \id(\e))$ of two nodes $\s$ and $\e$ corresponds to the length of the shortest path between them in the tree and is an upper bound on the routing.
Now, by Eq. \ref{eq:TD}, the sum of the length of the two coordinates is an upper bound on $\delta_{TD}(\id(\s), \id(\e))$.
As the length of a coordinate is equal to the level of the corresponding node in the tree, we indeed obtain
\begin{align}
\label{eq:rtd-bound}
\E(R^{TD}_{\s,\e})\leq \E(\delta_{TD}(\id(\s), \id(\e))) \leq \E(L_T(s))+\E(L_T(\e))= \calO\left( \frac{\trees}{q}\log n\right).
\end{align}
The last step follows from Lemma \ref{lem:tconst}.
Eq. \ref{eq:R-TD} follows because Eq. \ref{eq:rtd-bound} holds for all source-destination pairs $(\s,\e)$. 

In contrast, the proof for the common prefix length based similarities cannot build on previous results. 
Note that the change of the distance function does not affect the existence of a path with expected length at most $\E(L_T(s))+\E(L_T(\e))$ between source $\s$ and destination $\e$ in the tree. 
However, the routing might divert from that path when discovering a node with a longer common prefix length but at a higher depth.
For this reason, the sum of the expected levels is not an upper bound on the routing length.
Rather, whenever a node with a longer common prefix length is contacted, the upper bound of the remaining number of hops is reset to the expected level of that node in addition to the level of $\e$. 
In the following, we show that such a reset increases the distance in tree by less than $\frac{\trees}{q}$ on average.
The claim then follows because the number of resets is bound by the expected level of the destination. 
Eq. \ref{eq:R-CPL} follows by multiplication of the increased distance per reset and the number of resets.

More precisely, let $X_i$ give the tree distance between the $i$-th contacted node $v_i$ and the target $\e$.
Again, we cannot use the traditional approach for deriving the routing length because $X_i$ is not monotonously decreasing. Rather, we need to bound the number of times $Z_1$ that $X_i$ increases and the expected amount of increase $Z_2$.
Thus, the routing length $R^{CPL}_{\s,\e}$ from a source node $s$ to $\e$ is bound by
 \begin{align}
 \label{eq:z1z2}
 \E(R^{CPL}_{\s,\e}) \leq \E(L_T(s)) + \E(L_T(\e)) + \E(Z_1)\E(Z_2).
 \end{align}
The number of times $Z_1$ the common prefix length can increase is bound by the length of the target's coordinate and hence its level in $T$. So by Lemma \ref{lem:tconst}, 
\begin{align}
\label{eq:z1}
\E(Z_1) \leq \E(L_T(\e)).
\end{align}
The tree distance $X_i$ is potentially increased whenever a node with a longer common prefix length is contacted.
Yet, an upper bound on the expected increase is given by the difference in the levels $L_T(v_i)$ and $L_T(v_{i+1})$ minus $1$ due to the increased common prefix length.
Note that $v_i$ and $v_{i+1}$ are neighbors and hence the length of their shortest path to the root differs by at most $1$. Lemma \ref{lem:tconst} thus provides the desired bound on $\E(Z_2)$ 
 \begin{align}
 \label{eq:z2}
 \E(Z_2) \leq \E\left(L_T(v_i) - L_T(v_{i+1})\right) -1 = \frac{\trees}{q}.
 \end{align}
The desired bound can now be derived from Lemma \ref{lem:tconst}, Eqs. \ref{eq:z1z2}, \ref{eq:z1}, and \ref{eq:z2} under the assumption that the diameter of the graph and hence all shortest paths to the root scale logarithmically, i.e.,
 \begin{align}
 \label{eq:rcpl-bound}
 \E(R^{CPL}_{\s,\e}) \leq \E(L_T(s)) + \E(L_T(\e)) + \E(L_T(\e))\frac{\trees}{q} =
 \calO\left(\left(\frac{\trees}{q}\right)^2 \log n \right).
 \end{align}
As for the first part, Eq. \ref{eq:R-CPL} follows because Eq. \ref{eq:rcpl-bound} holds for all pairs $(\s,\e)$.
\end{proof}

The bounds for a virtual overlay lookup based on routing algorithm $\find$ follow directly from the fact that a DHT lookup requires $\calO(\log n)$ overlay hops with each hop corresponding to one route in the network embedding.

\begin{cor}
\label{cor:dht}
If the DHT used for the virtual overlay offers logarithmic routing, the communication complexity of routing algorithm
$\find$ is 
\begin{align*}
\E(DHT^{TD}) = \calO\left(\frac{\trees}{q}\log^2 n\right)
\end{align*}
for the diversity measure $\delta_{RP-TD}$ and
\begin{align*}
\E(DHT^{CPL}) = \calO\left(\left(\frac{\trees}{q}\right)^2 \log^2 n\right)
\end{align*}
 for diversity measure $\delta_{RP-CPL}$.
\end{cor}

\paragraph{Stabilization:} The stabilization complexity is required to stay polylog in the network size to allow for scalable communication and content addressing.
In the following, we hence give bounds for the self-stabilization of the network embeddings, the costs for the virtual overlay follow by considering the maintenance costs for DHT as suggested for general overlay networks and multiplying with the length of the routes between overlay neighbors.

\begin{theorem}
\label{thm:stabilize}
We assume the social graph $G$ to be of a logarithmic diameter and a constant average degree.
Furthermore, we assume the use of a the root election protocol with complexity $\calO(n \log n)$.
Then stabilization complexity $CS^\stab$ of the spanning trees constructed by Algorithm \ref{lem:tconst} with parameters $\trees$ and $q$ for one topology change is
\begin{align}
\label{eq:s}
\E(CS^\stab) = \calO\left( \trees \frac{\trees}{q} \log n\right).
\end{align}
\end{theorem}
\begin{proof}
We first consider the complexity for one tree. The general result then follows by multiplying with the number of trees $\trees$.
When a node joins an overlay with a constant average degree, the communication complexity of receiving and replying to all invitations is constant.
For a node departure, we consider  non-root nodes and root nodes separately.
If a any node but the root departs, the expected stabilization complexity corresponds to the number of nodes that have to rejoin $T$.
This number of nodes is equal to the number of descendants in a tree.
Hence, the expected complexity of a departure corresponds to the expected number of descendants. 
Consider that a node on level $l$ is a descendant of $l$ nodes, so that the expected number of descendants $D$ is given by
\begin{align*}
\E(D) = \frac{1}{|V|}\sum_{v \in V} \E(L_T(v)) = \calO\left( \frac{\trees}{q} \log n \right).
\end{align*}
If the root node leaves, the spanning tree and the embedding have to be re-established at a complexity of $\calO(n \log n)$.
As the probability for the root to depart is $1/n$, we indeed have
\begin{align*}
\E(CS^\stab) = \calO\left( \frac{\trees}{q} \log n \right) + \calO\left(\frac{1}{n} n \log n\right) = \calO\left( \frac{\trees}{q} \log n \right).
\end{align*}
\end{proof}

We have shown that the complexity of routing, content discovery, and stabilization is bound (poly-)log as required.

\subsection{Simulations}
\label{sec:voute-eff-sim}
In this section, we validate the above bounds and relate them to the concrete communication overhead for selected scenarios. 
We start by detailing our simulation model and set-up, followed by our expectation, the results and their interpretation. 

\paragraph{Model and Evaluation Metrics:}
In order to evaluate the efficiency, we consider the routing length and the stabilization complexity.
We express the stabilization costs in terms of the average number of coordinates that have to reassigned when a randomly chosen node leaves, i.e., the average number of descendants of a node.
The number of messages required for the assigning the new coordinates is at most two per assignment, namely the disconnected node registering at a new parent and receiving a new coordinate.
We conducted the study to determine how the number of trees, the tree construction algorithm, and the distance function affect routing and stabilization costs. 

We compared our results to those for Freenet, a virtual overlay $VO$, and the original PIE embedding.
The virtual overlay $VO$ combines the advantages of X-Vine and MCON by using shortest paths as tunnels in a Kademlia overlay like MCON 
but integrating backtracking in the presence of local optima and shortcuts from one tunnel to another like X-Vine.

\paragraph{Set-up:}
Due to space constraints, we restrict the presented results to one example network, namely the giant component of a community network from Facebook with 63392 users~\footnote{\url{http://konect.uni-koblenz.de/networks/facebook-wosn-links}}.

The spanning tree construction in Algorithm \ref{algo:Tconst} is parametrized by the number of trees $\trees \in \{1,2,3,5,7,10,12,15\}$, the acceptance probability $q=0.5$, and the selection criterion $W$ chosen to be either random selection (denoted \emph{DIV-RAND}) or preference of nodes at a low depth (denoted \emph{DIV-DEP}).
In addition, we consider a breadth first search for spanning tree construction (denoted \emph{BFS}).
Moreover, we consider the impact of the two distances $\delta_{TD}$ (denoted \emph{TD}) and $\delta_{CPL}$ (denoted \emph{CPL}).
The length of the return addresses was set to $\lengthAd=128$ and the number of bits per element was $b=128$, all $\tau=\trees$ embeddings were considered for routing.

For the virtual overlay used for content addressing, we chose a highly resilient recursive Kademlia \cite{heep2010r} with bucket size $8$ and $\alpha \in \{1,3\}$ parallel look-ups. Because routing table entries are not uniquely determined by Kademlia identifiers, the entries were chosen randomly from all suitable candidates. 

We parametrized the related approaches as follows. 
For simulating Freenet, we executed the embedding for $6,000$ iteration as suggested in \cite{sandberg2006distributed} and then routed using a distance-directed depth-first search based only on the information about direct neighbors.
The routing and stabilization complexity of the original PIE embedding is equal to the respective quantities of our algorithm for $\trees=1$, the distance function $\delta_{TD}$ and routing without the use of backtracking. 
In order to better understand the results of the comparison, we simulate the virtual overlay $VO$ using the same Kademlia overlay as for our own approach but replacing the tree routing by tunnels corresponding to the shortest paths between overlay neighbors. 
So, we parametrized the related approaches by either using the proposed standard parameters or selecting parameters that are suitable for comparison because they corresponds to the same degree of redundancy as the parametrization of our own approach.

All results were averaged over 20 runs. They are displayed with $95$\% confidence intervals.
Each run consisted of $100,000$ routing attempts for a randomly selected source-destination pair.

\paragraph{Expectations:}
We expect that the routing length decreases with the number of embeddings, because the number of available routes and thus the probability to discover the shortest route in one embedding increases. 
In general, the routing length is directly related to the tree depth and should thus be lower for \emph{BFS} and \emph{DIV-DEP}.

Similarly, we expect a higher stabilization overhead for trees of a higher depth as the expected number of descendants per node increases.
Thus, the number of nodes that need to select a new parent should be higher for \emph{DIV-RAND} than for \emph{DIV-DEP} and \emph{BFS}.

In comparison to the existing approaches, our approach should enable shorter routes between pairs of nodes than both Freenent and VO.
As shown above, we achieve a routing complexity of $\calO\left(\log n\right)$ whereas the related work achieves at best routes of polylog length.
However, our routes for content discovery should be slightly longer than in VO. VO utilizes the same DHT routing but uses shortest paths rather than the longer tree routes.

\paragraph{Results:}

\begin{figure*}[ht]
\subfloat[$\routeN$]{\label{fig:rt}\includegraphics[width=0.33\linewidth]{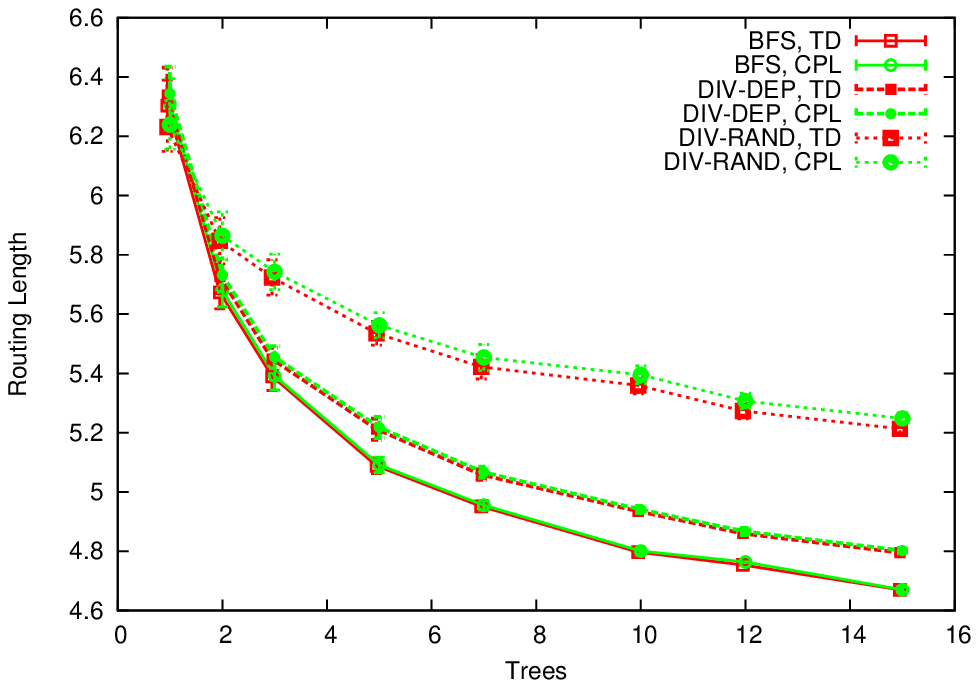}}
\hfill
\subfloat[$\find$]{\label{fig:kad}\includegraphics[width=0.33\linewidth]{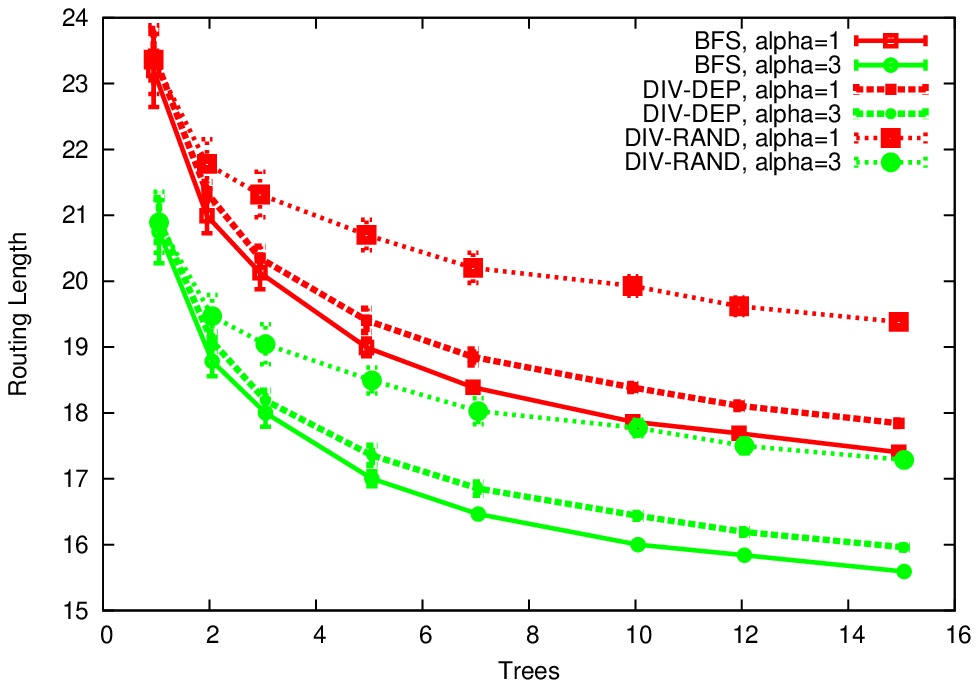}}
\hfill
\subfloat[Stabilization]{\label{fig:main}\includegraphics[width=0.33\linewidth]{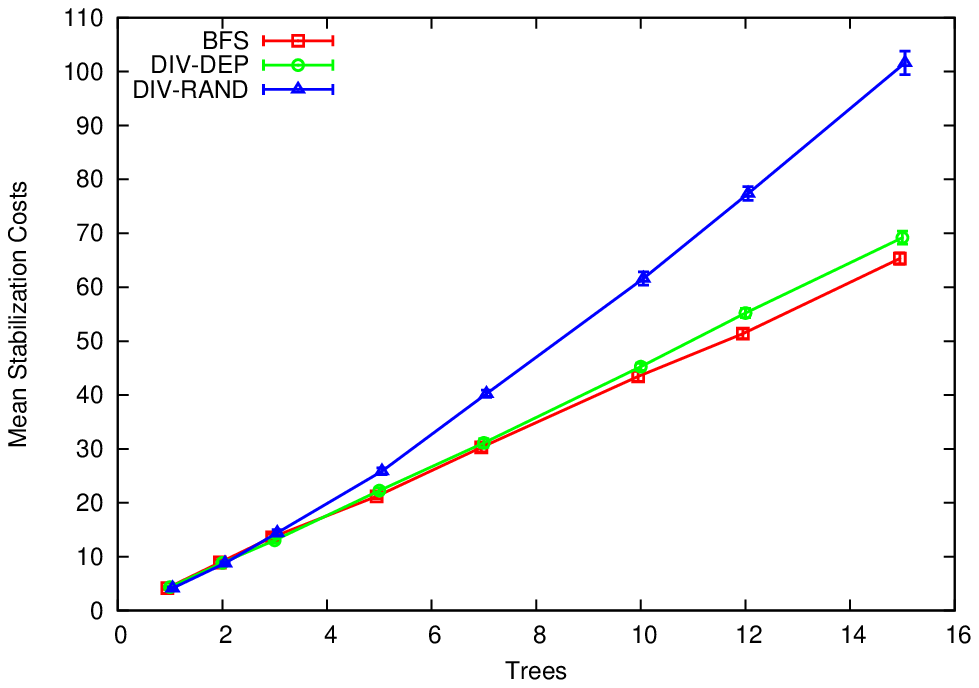}}
\caption{Impact of number of embeddings $\trees$, tree construction, and distance function on routing length for a) tree routing and b) Kademlia lookup with  degree of parallelism $\alpha$; related approaches result in routing lengths of $14$ (virtual overlay $VO$) and close to $10,000$ (Freenet), and c) stabilization overhead }
\label{fig:static}
\end{figure*}

The impact of the three parameters, number of trees, tree construction, and distance on the routing length confirms our expectations.
First, the results indicate that the tree construction, in particular the number of trees, is the dominating factor for the routing length. So, the routing length decreased considerably if multiple embeddings were used because the shortest route in any of the trees was considered. 
Second, preferring parents closer to the root, i.e., using \emph{BFS} or \emph{DIV-DEP}, produced shorter routes in the tree and hence reduced the routing length.
Third, in comparison to the tree construction, the choice of a distance function had less impact.
For \emph{BFS} or \emph{DIV-DEP}, the advantage of \emph{TD} over \emph{CPL} was barely noticeable, whereas the difference for \emph{DIV-RAND} was still small but noticeable. 
In order to understand this difference, note that \emph{CPL} is expected to lead to longer routes.
The reason for the longer routes lies in forwarding the request to neighbors at a higher depth, which might have a long common prefix but are nevertheless at a higher distance from the destination due to their depth. 
For \emph{BFS} or \emph{DIV-DEP}, the difference of the depth of neighbors was generally small because neighbors at a lower depth were preferably selected as parents. 
In contrast, \emph{DIV-DEP} allows for larger differences in  depth. Hence there is a higher probability to increase the tree distance by selecting a neighbor with a longer common prefix length but at a high depth.
All in all, the routing length varied between $4.67$ (\emph{BFS}, $\trees=15$, \emph{TD}) and $6.24$ (\emph{DIV-RAND}, $\trees=1$, \emph{CPL}) hops, as displayed in Figure \ref{fig:rt}.
In summary, the use of  multiple embeddings indeed reduced the routing length considerably.

The performance of the DHT lookup in the virtual overlay directly related to the previous results (cmp. Fig. \ref{fig:kad} for the distance under {\em TD}).
The overhead for the discovery of a randomly chosen Kademlia ID, stored at the node with the closest ID in the overlay, varied between $15.56$ and $24.25$ hops in the F2F overlay, at around 4 hops in the virtual overlay.

By Theorem \ref{thm:stabilize}, the stabilization complexity was expected to increase at most quadratic with the number of trees.
Indeed, Figure \ref{fig:main} supports this fact for \emph{DIV-RAND}.
The increase for \emph{BFS} and \emph{DIV-DEP} was even only linear and slightly super-linear, respectively.
Note that the quadratic increase is due to the raising average depth of additional trees generated by Algorithm \ref{algo:Tconst}. 
With the goal of achieving diverse spanning trees, nodes select parents at a higher depth.
However, the average number of descendants increases with the depth,
because a node at depth $l$ is a descendant of $l$ nodes.
Due to the stabilization complexity corresponding to the number of the departing node's descendants,
the stabilization overhead was  higher for \emph{DIV-RAND} and \emph{DIV-DEP} than for \emph{BFS}.
More precisely, \emph{BFS} constructs all $\trees$ trees independently, so that the average depth of each tree is independent of the number of trees.
The stabilization complexity per tree thus remains constant. 
\emph{DIV-DEP}, aiming to balance diversity and short routes, causes stabilization overhead between the two former approaches, but performed closer to \emph{BFS} (this similarity also held for the routing length).
More concretely, the average stabilization overhead for a departing node was slightly below $4.5$ for a single tree.
For $\trees=15$ it increased to $65$ (\emph{BFS}), $69$ (\emph{DIV-DEP}), and more than $101$ (\emph{DIV-RAND}).
In contrast to a complete re-computation of the embedding requiring at least $n=63392$ messages, the stabilization overhead is negligible.

For the related approaches, we found a routing length of $9403.1$ for Freenet, $16.11$ for VO with $\alpha=1$, and $14.07$ for VO with $\alpha=3$. 
Furthermore, the shortest paths are on average of length $4.31$, meaning that our routing length of $4.67$ is close to optimal.
So, routing between nodes in the tree required less than half the overhead of  state-of-the-art approaches.
Routing in the virtual overlay, requiring at best less than $16$ hops in our scheme, was slightly more costly in our approach than in VO due to the inability of the tree routing to guarantee shortest paths between virtual neighbors.

A straight-forward comparison of the stabilization overhead was not possible. 
Since Freenet stabilizes periodically, there is no overhead directly associated with a leaving node. 
In case of virtual overlays, VO uses flooding for stabilization, which is clearly more costly. 
Other overlays such as X-Vine use less costly stabilization but stabilization and routing overhead are unstable and increase over time as shown in \cite{roos2015impossibility}, so that it is unclear which state of the system should be considered for a comparison.
In order to nevertheless give a lower bound on the stabilization overhead, we computed the number of tunnels that needed to be rebuild in VO. On average, $477.35$ tunnels corresponding to shortest paths were affected by a departing node. 
If a tunnel is repaired by routing in the Kademlia overlay like in X-Vine, the stabilization overhead per tunnel corresponds to routing a request and the corresponding reply, i.e., for tunnels corresponding to shortest paths at least $2 \cdot 14=28$ messages, resulting in a lower bound on more than $10,000$ messages per node departure. 
The above stabilization algorithm is unable to maintain short routes, such that the actual overhead of stabilization in virtual overlay is even higher than the above lower bound.

\paragraph{Discussion:}
Our simulation study validates the asymptotic bounds. 
Indeed, the routing length and thus the routing complexity for messaging is very low, improving on the state-of-the-art by more than a factor of $3$.
The stabilization complexity is similarly low if the number of trees is not too high. Even for $\trees=15$ trees, the number of involved nodes is generally well below 100, which still improves upon virtual overlays such as VO, the most promising state-of-the-art candidate. 
Only content discovery in form of a DHT lookup was slightly more costly in our approach than in VO, which we consider acceptable given the considerable advantage  with regard to all other metrics.

\paragraph{}We have considerably improved the efficiency of F2F overlays. In the following, we show that we also mitigated their vulnerability to failures and attacks.

%% file: eval-resilience.tex
\section{Robustness and Censorship-Resistance}
\label{sec:voute-eval-resi}

In this section, we consider the robustness and resilience to censorship of VOUTE.
Note that the evaluation of the censorship-resilience requires a specification of the modified stabilization algorithm
$\stab'$, which refer to as \emph{attack strategy} in the following. 
After deriving two attack strategies, we subsequently present our theoretical and simulation-based evaluation.

We express our results in terms of node coordinates and distances $\delta_{TD}$ and $\delta_{CPL}$ rather than the corresponding diversity measures. The use of distances simplifies the notation as we do not need to apply a hash cascade for the comparison of coordinates and return addresses.
As the routing paths are chosen identical for both coordinates and return addresses, the results are equally valid for return addresses.

\subsection{Attack Strategy}
We first describe our attack strategies and then comment on additional strategies and our reasons on selecting 
In order to model secure and insecure root selection protocols, we consider two realizations of \emph{ATT-RAND} and \emph{ATT-ROOT}.
In the following, assume that one attacker node has established $x$ links to honest nodes and now aims to censor communication. 

For secure spanning trees, the adversary $A$ is unable to manipulate the root election.
Nevertheless, $A$ can manipulate the subsequent embedding.
The attack strategy \emph{ATT-RAND} assigns each of its children a different
random prefix rather than the correct prefix. In this manner,
routing fails because
nodes in the higher levels of the tree do not recognize the
prefix. So, the impact of the attack is increased in comparison to a random failure.

In contrast, if the adversary $A$ can manipulate the root election protocol, \emph{ATT-ROOT} manipulates the root election in all spanning trees such that $A$ becomes the root in all trees.
Under the assumption that the root observes the maximal number of requests, the attack should result in a high ratio of failed requests.

\paragraph{} Now, we shortly comment on some further attack strategies we choose not to implement and give reasons for our decision not to do so.

First, note that in the original PIE embedding, assigning the same coordinate to two children is another attack strategy. In contrast to the above strategy, the routing can then fail even if the attacker is not involved in forwarding the actual request because the node coordinates are not unique and thus the request might end up at a different node than the receiver. 
In the modified embedding, the child decides on the last element of the coordinate. 
Hence, the attacker can only assign a node $w$ the coordinate of another node $v$ as a prefix, so that the two nodes appear to be close but are indeed not. However, upon realizing that $w$ does not offer a route to $v$, the routing algorithm backtracks, so that this attack strategy merely increases the routing complexity but not the success ratio.
Thus, we do not consider it here.  

Second, recall from Section \ref{sec:req} that the attacker can also generate an arbitrary number of identities whereas the above attack strategies only rely on one identity. In the following, we argue that without additional knowledge, the use of additional identities in the tree does not improve the strength of the attack.

 \emph{ATT-RAND} actually simulates different virtual identities by providing fake distinct prefixes to all children.
 Indeed, in practice, it might be wise to indeed use distinct physical nodes because it minimizes the risk of detection if two neighbors realize that they are connected to the same physical node but received different prefixes.

For \emph{ATT-ROOT}, the attacker might have to create (virtual) identities in order to manipulate the root election. 
As soon as $A$ is the root in each tree, multiple identities could be used to provide prefixes of different lengths. 
However, if a neighbor $u$ of $A$ receives a long prefix from $A$, there is a high chance that $u$ and potential descendants of $u$ choose different parents seemingly closer to the root.
Thus, in expectation a large number of nodes joins those subtrees rooted at a neighbor of $A$ with a short prefix. As routing within such a subtree does not require to forward a request from $A$, $A$'s impact is likely to be reduced. 
Hence, without concrete topology knowledge, the insertion of additional virtual identities (corresponding to prefixes of different lengths) does usually not present an obvious advantage for $A$.

\subsection{Theoretical Analysis}

We present two theoretical results in this section. First, we characterize the backtracking algorithm more closely.
Second, we show that the censorship-resistance is improved by using the distance $\delta_{CPL}$ rather than
$\delta_{TD}$.

Throughout this section, let $\routeTD$ and $\routeRAPCPL$ denote Algorithm \ref{algo:route} with
distance $\delta_{TD}$ and $\delta_{CPL}$, respectively. Furthermore, let $\mathbf{GR}^{TD}$ and $\mathbf{GR}^{CPL}$  denote the
corresponding standard greedy routing algorithms, which terminate in local optima with regard to the distance to the destination's coordinate.  
Let $Succ^\route$ denote the success ratio of a routing algorithm $\route$. 
We are considering the success ratio for one embedding. The overall success ratio is improved as it is the combined
success ratio of all embeddings. 

\begin{lemma}
\label{thm:voute-back}
We have that 
\begin{align}
\label{eq:backbetter}
\begin{split}
\E\left(Succ^{\routeTD}\right)&\geq \E\left(Succ^{\mathbf{GR}^{TD}}\right)\\
\E\left(Succ^{\routeRAPCPL}\right)&\geq \E\left(Succ^{\mathbf{GR}^{CDF}}\right)
\end{split}. 
\end{align}
Furthermore, Algorithm \ref{algo:route} is successful if and only if there exists a greedy path of responsive nodes according to its distance
metric $\dist$.
\end{lemma}
\begin{proof}
Eq. \ref{eq:backbetter} follows because Algorithm \ref{algo:route} is identical to the standard greedy algorithm until the
latter terminates. Then, Algorithm \ref{algo:route} continues to search for an alternative, possible increasing the success ratio.

For the second part, recall that a greedy path is a path $p=(v_0, \ldots , v_l)$ such that the distance to the destination $v_l$ decreases in each step, i.e., $\delta(\id(e), \id(v_i)) < \delta(\id(e), \id(v_{i-1}))$ for all $i=1..l$ and a distance $\delta$.
Assume Algorithm \ref{algo:route} does not discover a route from the source $v_0=\s$ and $v_l=\e$ despite the existence of a greedy path $p=(v_0, v_1, \ldots , v_{l-1}, v_l)$ of responsive nodes. Let $V_R$ be the set of nodes that forwarded the request according to Algorithm \ref{algo:route}, and let 
$j=\max \{i: v_i \in V_R\}$.
Then the neighbor of $v_{j+1}$ did not receive the request despite being closer to $\e$ than $v_j$. Though $v_j$ might have a neighbor  $w$ closer to $\e$ than $v_{j+1}$, the request is backtracked to $v_j$ if forwarding to $w$ does not result in a route to the destination.
Routing only terminates if either a route is found or $v_j$ has forwarded the request to all closer neighbors, including $v_{j+1}$.
Thus, Algorithm \ref{algo:route} cannot fail if a greedy path exists.
In contrast, if there are not any greedy paths from $\s$ to $\e$, any path  $p=(v_0, v_1, \ldots , v_{l-1}, v_l)$ with $v_0=\s$ and $v_l=\e$ contains a pair $(v_{i-1},v_{i})$ with
$\delta(\id(e), \id(v_i)) \geq \delta(\id(e), \id(v_{i-1}))$. Thus, Algorithm \ref{algo:route} does not forward the request to $v_{i}$ and hence does not discover a path from $\s$ to $\e$. 
It follows that indeed  Algorithm \ref{algo:route} is successful if and only if a greedy path of responsive nodes exists. 
\end{proof}

Now, we use Lemma \ref{thm:voute-back} to show that using a common prefix length based distance generally enhances the censorship-resistance.
\begin{theorem}
\label{thm:voute-betterAS}
Let $A$ be an attacker applying either \emph{ATT-RAND} or \emph{ATT-ROOT}. Then for all distinct nodes $\s, \e \in V$
\begin{align}
\label{eq:succ_implies}
Succ^{\routeRAPCPL}_{\s,\e}=0 \implies Succ^{\routeTD}_{\s,\e}=0,
\end{align}
i.e., if $\routeRAPCPL$ does not discover a route between $\s$ and $\e$, then $\routeTD$ does not discover a route.
However, the converse does not hold. 
In particular, 
\begin{align}
\label{eq:Esucc}
\E\left(Succ^{\routeRAPCPL}\right) \geq \E\left(Succ^{\routeTD}\right).
\end{align}
\end{theorem}
\begin{figure}[ht]
\centering
\includegraphics[width=0.7\linewidth]{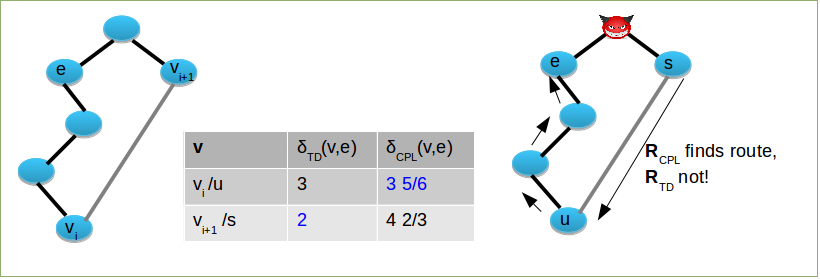}
\caption[Illustration of Resistance Proof]{Illustrating the proof of Theorem \ref{thm:voute-betterAS}: left: $v_{i+1}$ is closer to destination
$\e$ than $v_i$ for distance $\delta_{TD}$ but not for $\delta_{CPL}$; right: pair $(\s, \e)$ for which $\routeRAPCPL$ is successful as $\s$ forwards to $u$,
but $\routeTD$ is not successful because $\s$ forwards to the attacker.}
\label{fig:eval-resi}
\end{figure} 
\begin{proof}
We prove the claim by contradiction. Assume that there is pair $\s, \e$ such that the algorithm $\routeTD$
terminates successfully while $\routeRAPCPL$ does not.
Let $p=(v_0, v_1, \ldots , v_l)$ with $v_0=\s$ and $v_l=\e$ denote the discovered route.
By Lemma \ref{thm:voute-back}, $p$ is a greedy path for distance $\delta_{TD}$  but not for $\delta_{CPL}$.
In other words, there exists $0\leq i < l$ such that i) $\delta_{TD}(\id(v_{i+1}), \id(\e)) < \delta_{TD}(\id(v_i),\id(\e))$ and
ii) $\delta_{CPL}(\id(v_{i+1}), \id(\e)) \geq \delta_{CPL}(\id(v_i),\id(\e))$.
By the definitions of both distances in Eq. \ref{eq:TD} and Eq. \ref{eq:CPLD}, this implies
that $cpl\left(\id(v_{i+1}), \id(\e)\right)< cpl\left(\id(v_{i}), \id(\e)\right)$ and $|\id(v_{i+1})|<|\id(v_{i})|$.
In other words, $v_{i+1}$'s coordinate has a lower common prefix length to $\id(\e)$ and is shorter than $\id(v_i)$.
The right side of Figure \ref{fig:eval-resi} displays an example. 

We base our contradiction upon the following observation concerning routes in trees. 
Consider the \emph{tree route} between two nodes, i.e., the path between them using only tree edges.
Along the tree route, 
the common prefix length stays constant until the least common ancestor is reached and then increases.
Now, if $\id(v_{i+1})$ has a shorter common prefix with $\id(\e)$ than $\id(v_i)$, $v_{i+1}$ is not 
contained in the tree route. 
Furthermore, as the routing algorithm $\routeRAPCPL$ does not successfully discover a route, the attacker has to control one node
on the tree route. 

We can use the above observation to establish contradictions for both \emph{ATT-RAND} and \emph{ATT-ROOT}. 
Note that if the common prefix length decreases when forwarding to $v_{i+1}$, we need to have $cpl\left(\id(v_{i}), \id(\e)\right)>0$.
For the attack strategy \emph{ATT-RAND}, the attacker on the tree route is either an ancestor of $v_i$ or of $\e$.
However, the attacker replaces the prefixes of all its children and hence descendants, so that the perceived common prefix length of
$v_i$' and $\e$'s coordinates should be 0 unless there exists an attacker-free tree route. 
This is a clear contradiction. 
Similarly, if $cpl\left(\id(v_{i}), \id(\e)\right)>0$, $v_i$ and $\e$ have a common ancestor aside from the root. In particular, the tree route
does not pass the root. When applying \emph{ATT-ROOT}, the only attacker is the root, which again contradicts that there is an attacker
on the tree route. 
Thus, we have shown by contradiction that $\routeTD$ only succeeds if $\routeRAPCPL$ does.

Thus, we have shown that indeed Eq. \ref{eq:succ_implies} holds. Eq. \ref{eq:Esucc} is a direct consequence as it averages over
all source-destination pairs and systems. It remains to show that the converse of Eq. \ref{eq:succ_implies} does not hold.
In other words, there exist instances when $\routeRAPCPL$ terminates successfully while $\routeTD$ fails. 
We display such an example in Figure \ref{fig:eval-resi}. 
\end{proof}

While we can show that our enhancements are indeed enhancements, our theoretical analysis does not provide any absolute bounds on the success ratio. In particular, we cannot compare our success ratio to that of virtual overlays.

\subsection{Simulations}
We utilized the simulation model and set-up from Section \ref{sec:voute-eval-effi} for evaluating the efficiency and extended it to include the methodology for robustness and censorship-resistance.
So, we simulate the robustness of an overlay by subsequently selecting random failed nodes. In each step, we select a certain fraction of additional failed nodes and then determine the success ratio.  
Furthermore, we evaluate attacks using the two attack strategies \emph{ATT-RAND} and \emph{ATT-ROOT} described above. 
We first establish the overlay applying the respective attack strategy and then execute the routing for randomly selected source-destination pairs of responsive nodes. 

We compare our results to the virtual overlay VO, described in Section \ref{sec:voute-eval-effi}. 
Our attacker on VO does not manipulate the tunnel establishment but merely drops requests. 
Recall that routing in VO relies on a Kademlia DHT such that neighbors in the DHT communicate via a tunnel of trusted links.
The routing between two DHT neighbors thus fails if the attacker is contained in the tunnel. 
However, if routing between two overlay neighbors fails, the startpoint of the failed tunnel can attempt to select a different overlay neighbor
as long as it has one neighbor closer to the destination. 
We further enhance the success ratio of VO by optionally allowing backtracking in the DHT. 
In addition, we also allow for shortcuts, i.e., rather than following the tunnel to its endpoint, nodes on the path can change to a different tunnel with an endpoint closer to the destination.
Thus, we maximize the chance of successful delivery in VO by backtracking and shortcuts in addition to the use of non-strategic attacker.

\paragraph{Set-up:} We used the embedding and routing algorithms as parametrized in Section \ref{sec:voute-eval-effi}.

In order to evaluate the robustness,  we removed up to $50$\% of the nodes in steps of $1$\%. During the process of removing nodes, individual nodes inevitably became disconnected from the giant component, so that routing between some pairs was no longer possible. For this reason, we only considered the results for source-destination pairs in the same component. Our results are presented for $1$, $5$, and $15$ trees only.

The number of edges $x$ controlled by the adversary $A$ were chosen as $x=2^i \times \lceil \log_2 n \rceil$ with $0\leq i \leq 6$ and $\lceil \log_2 n \rceil =16$. So, up to $1,024$ attacker edges were considered. In particular, $x=1024> \frac{\sqrt{n}}{\log n}$, a common asymptotic bound on the number of edges to honest nodes considered for Sybil detection schemes \cite{danezis2009sybilinfer}. 
For quantifying the achieved improvement, we compared our approach to the resilience of the original PIE embedding and routing, i.e., $1$ tree, $\delta_{TD}$, and no backtracking.

For VO, we used a degree of parallelism of $\alpha=1$. 
Since backtracking was applied, all values of $\alpha > 0$ resulted in the same success ratio, because regardless of the value of $\alpha$, the routing succeeded if and only if a greedy path in the virtual overlay existed. 
Thus, restricting our evaluation to $\alpha=1$ did not impact our results with regard to the success ratio. 

We averaged the results over $20$ runs with $10,000$ source-destination pairs each. Results are presented with $95 \%$ confidence intervals.

\paragraph{Expectations:}
We expect that the use of backtracking already increases the success ratio considerably for $\trees=1$.
However, for large failure ratios or a large number of attacker edges, the single-connected nature of the tree should result in a low success ratio. 
By using multiple trees, we expect to further increase the success ratio until close to 100\% of the paths correspond to a greedy path and hence a route in at least one embedding.  

For the robustness to failures, the original distance function \emph{TD} should result in a higher success ratio than \emph{CPL} because of its shorter routes, as seen in Section \ref{sec:voute-eff-sim}, and thus lower probability to encounter a random failed node.  
However, by Theorem \ref{thm:voute-betterAS}, \emph{CPL} increases the success ratio in contrast to the original distance.

Our first attack strategy, \emph{ATT-RAND}, should not have a strong impact as the fraction of controlled edges is low and the attacker usually does not have an important position in the trees.
In contrast, we expect many requests to be routed via the root, so that at least for a low number of trees, \emph{ATT-ROOT} should be an effective attack strategy. 

In comparison, our attack on VO does not enable the attacker to obtain a position of strategic importance, so that the impact of the attack should be much less drastic than \emph{ATT-ROOT}. 
However,  communication between DHT neighbors relies on one tunnel whereas tree embeddings provide multiple routes.
Thus, when using multiple diverse trees, we expect our approach to be similarly effective as VO, possibly even more effective.

\paragraph{Results:}
While the results verified our expectations with regard to the advantage of the distance \emph{TD} for random failures and of \emph{CPL} for attacks,
the observed differences between the two distances were negligible, i.e., less than $0.1$\%. Hence, we present the results for \emph{CPL} in the following with the exception of the results for the original PIE embedding.   

We start by evaluating the robustness to random failures.
The results, displayed in Figure \ref{fig:failures}, indicate that the use of multiple embeddings considerably increased the robustness.
The success ratio for $\trees=1$ was low, decreasing in a linear fashion to less than $30$\% for a failure ratio of $50$\%.
In contrast, for $\trees=15$, the success ratio exceeded $90$\%.
Though the number of embeddings was the dominating factor, the tree constructing algorithm also strongly influenced the success ratio.
For $\trees > 1$,  aiming to choose distinct parents improved the robustness to failures because of the higher number of distinct routes. 
For example, when routing in $5$ parallel embeddings, the success ratio was above $80$\% for \emph{DIV-RAND}.
In contrast, \emph{BFS} had a success ratio below $70$\%. 
In summary, the robustness to failures was extremely high for multiple embeddings, enabling a success ratio of more than $95$\% for up to $20$\% failed nodes. The robustness was further increased by using \emph{DIV-RAND} or \emph{DIV-DEP} rather than \emph{BFS}, showing that even such relatively simple
 schemes can achieve a noticeable improvement.

Now, we consider the censorship-resistance for $x=16$ attacking edges, as displayed in Figure \ref{fig:attacks}.
If the adversary $A$ was unable to manipulate the root selection, the success ratio was only slightly below $100$\%. Even if $\trees=1$, more than $99.5$\% of the routes were successfully discovered. The high resilience against \emph{ATT-RAND} was to be expected, considering that the attack was only slightly more severe than failure of one random node.
If the attacker was able to become the root in all trees, the success ratio dropped to about $93$\% for $\trees=1$. 
However, with multiple trees, the ratio of \emph{ATT-ROOT} was close to $100$\%.
The impact of the tree construction was small but noticeable.
So, \emph{BFS} generally resulted in a slightly lower success ratio.
Hence, by using multiple embeddings and backtracking, the resilience to an adversary that can establish only $\lceil \log_2 |V| \rceil$ is such that nearly all routes are successfully found. 

\begin{figure*}[ht]
\centering
\subfloat[Robustness]{\label{fig:failures}\includegraphics[width=0.33\linewidth]{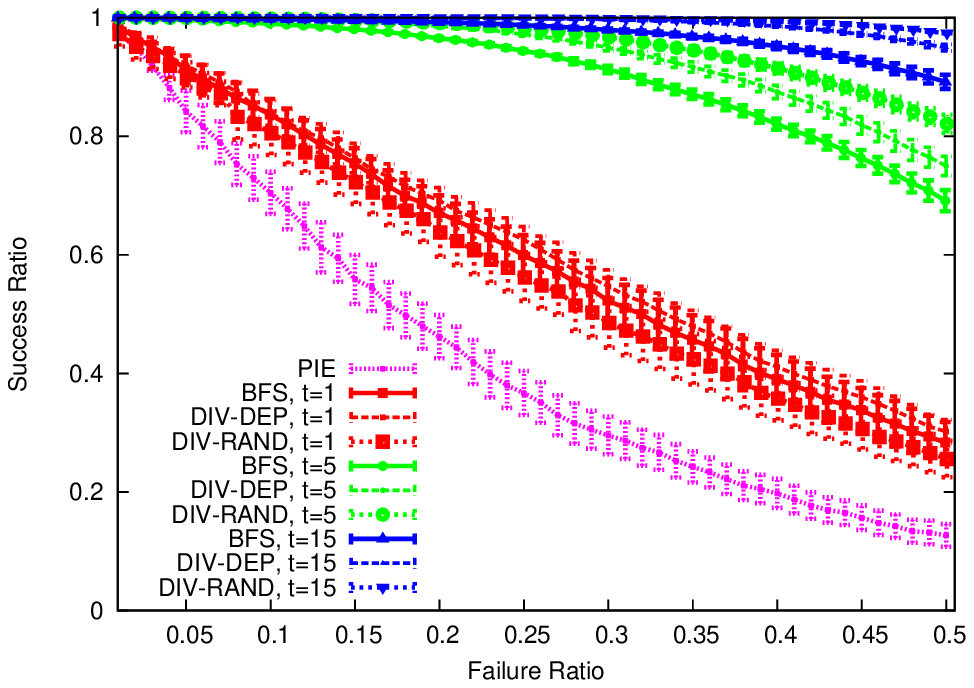}}
\hfill
\subfloat[Attacks: 16 Edges]{\label{fig:attacks}\includegraphics[width=0.33\linewidth]{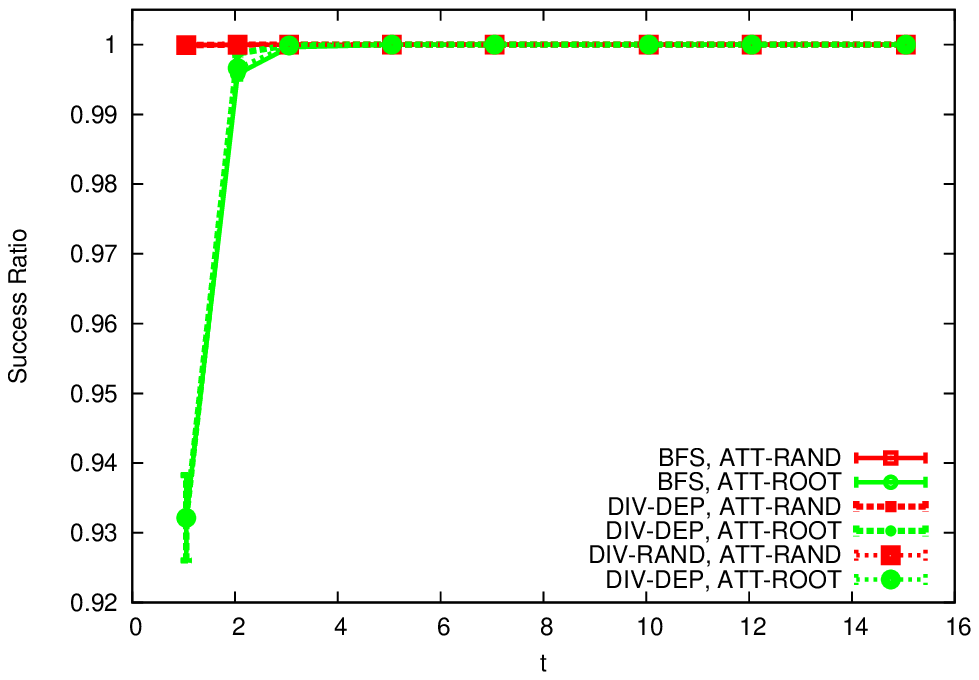}}
\hfill
\subfloat[Attacks: Up to 1024 edges]{\label{fig:attacksEdges}\includegraphics[width=0.33\linewidth]{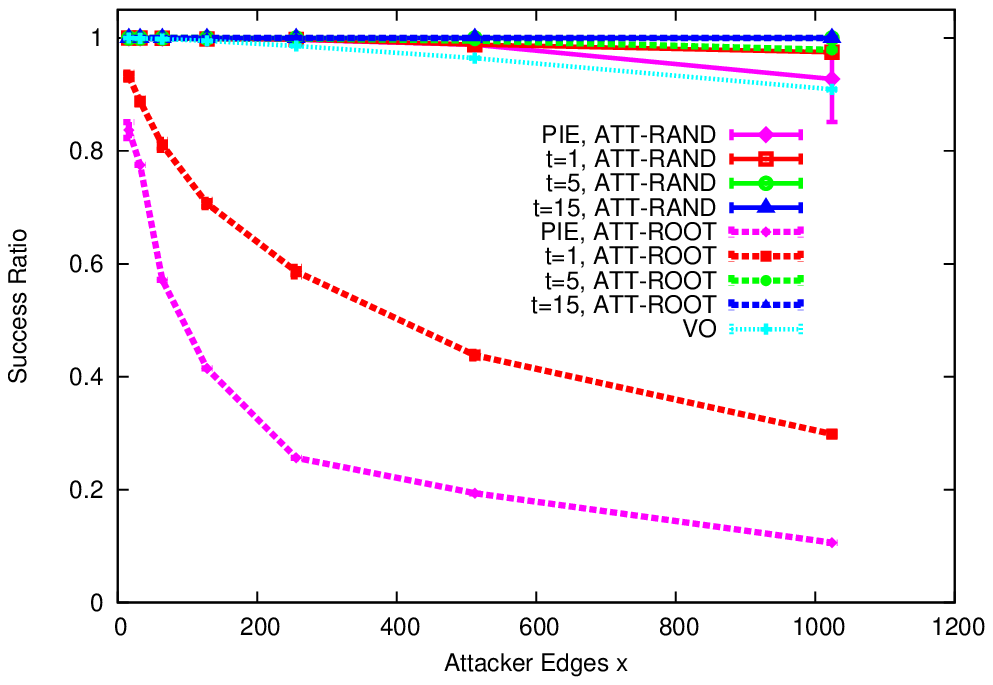}}
\caption[Attack Resilience of VOUTE]{a) Robustness to failures for distance \emph{CPL} and b),c) Censorship-Resistance of tree routing for distance \emph{CPL} to adversaries which are either able to undermine the root election (\emph{ATT-ROOT}) or are unable to do so (\emph{ATT-RAND}) for b) $x=16$ attacking edges, and c) up to $1,024$ attacking edges and tree construction \emph{DIV-DEP}}
\label{fig:resi}
\end{figure*}

For an increased number of attacking edges $x$, the success ratio remained close to $100$\% when more than one tree was used for routing, as displayed in Figure \ref{fig:attacksEdges} for \emph{DIV-DEP}. However, for one tree, the success ratio decreased drastically if an attacker could undermine the root selection. For $x=1024$, i.e., if the attacker controlled edges to roughly $1.7$\% of the nodes, the success ratio for $\trees=1$ decreased to slightly less than $30$\%.
In contrast, if $\trees=5$ or $\trees=15$,  the success ratio was still $97.9$ or $99.9$\%, respectively.

In order to quantify the improvements provided by our resilience enhancements, we compared the results for our approach with the PIE embedding.
As can be seen from Figure \ref{fig:failures}, the success ratio dropped much more quickly for PIE than for the improved approaches.
For an adversary with $x=16$ connections to honest nodes, PIE suffered from more than twice the numbers of failed requests than the remaining systems (Figure \ref{fig:attacksEdges}) because it relies on only one tree and does not apply backtracking. 
When increasing the number of attacker edges, the success ratio dropped further to less than $15$\% for $x=1024$.
Our approach achieved more than twice the success ratio even for $\trees=1$.

In contrast to PIE, VO exhibited a rather high success ratio as displayed in Figure \ref{fig:attacksEdges}.
VO's advantage in contrast to $\trees=1$ holds despite VO's longer routes (see Section \ref{sec:voute-eff-sim}). 
The reason for VO's lower vulnerability lies in the absence of strategic manipulation. 
While greedy embeddings allow  the attacker to assume an important role, our attacker in VO does not attract a disproportional fraction of traffic.
However, establishing multiple trees ensures that the role of the root is effectively mitigated, so that the censorship-resilience of VO is slightly lower than VOUTE's resilience for $5$ or more parallel embeddings. 

\paragraph{Discussion:}
We have shown that multiple embedding and backtracking enable high resilience, outperforming state-of-the-art approaches.
Here, we focused on node-to-node communication. 
While content retrieval results in longer routes, we expect the success ratio to be similar as backtracking in the DHT allows
the use of multiple paths.
In addition, the number of replicas per content can be adjusted to increase the success ratio.

While Theorem \ref{thm:voute-betterAS} shows the advantage of \emph{CPL} in the presence of failures, the actual advantage is negligible,
so that it seems more sensible to use the original distance \emph{TD} due to its higher efficiency.

\paragraph{}In summary, our enhancements to the robustness and censorship-resistance were both needed and highly effective.

%% file: eval-anonymity.tex
\section{Anonymity and Membership-Concealment}
\label{sec:voute-eval-ano}

We show that our return addresses provide plausible deniability.

\begin{theorem}
\label{thm:partial}
Let $u$ be a local attacker, which is aware only of its direct neighbors $N_u$ in the social graph, and
let $y=(y_1, \ldots , y_{\trees})$ with routing information $\tilde{k}=(\tilde{k}_1, \ldots , \tilde{k}_\trees)$  be a vector of return addresses generated by Algorithm \ref{algo:rap} for the node $u_y$. 
Then $u$ cannot identify $\e_y$ with absolute certainty using a polynomial-time algorithm $A$, i.e., $P(A(y, \tilde{k})=\e_y) < 1$ for all return address vectors $y$.
So, we guarantee possible innocence with regard to both sender and receiver anonymity in the absence of identifying side channel information such as timing analysis.
\end{theorem}
\begin{figure}[ht]
\centering
\includegraphics[width=0.9\linewidth]{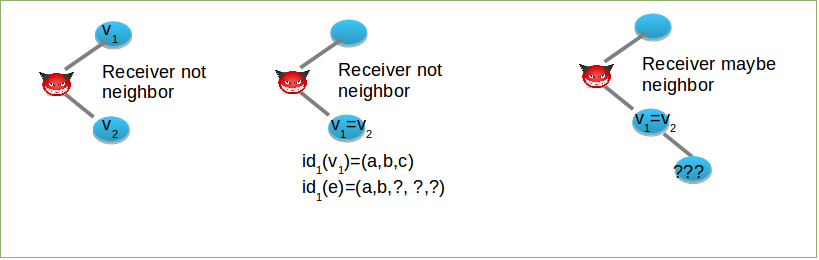}
\caption[Illustration of Anonymity Proof]{Illustrating the proof of Theorem \ref{thm:partial}: Let $v_1$ and $v_2$ be the neighbors with the closest coordinates to the receiver's coordinates $\id_1(\e)$ and $\id_2(\e)$ in the first and second embedding, respectively. 
The attacker can only infer if neighbors are not the receiver $\e$ but can not tell if they are. In particular, the attacker $A$ knows the common prefix of the receiver's coordinates and the coordinates of its neighbors but not the remaining elements of the coordinate, as indicated by the ?s in the coordinate. 
In the first scenario, $v_1 \neq v_2$ shows that the receiver is not a neighbor. 
In the second scenario, $A$ can infer that the third element of the receiver's coordinate $\id_1(\e)$ is not $c$ and hence $v_1$ is not the receiver.
In the last scenario, $A$ is unable to tell if a neighbor is indeed the receiver or if a child of the neighbor is the receiver. }
\label{fig:eval-ano}
\end{figure} 
\begin{proof}
We start by considering receiver anonymity. We consider three cases and show for each case that either i) the attacker can determine that the receiver is not a neighbor but cannot infer the coordinate of the actual receiver or ii) the attacker remains uncertain if the receiver is a neighbor or a neighbor's descendant. We illustrate the cases in Figure \ref{fig:eval-ano}.
Throughout the proof, let $v_i \in N_{u}$ be the closest neighbor of $u$ to $y_i$ for $i = 1, \ldots , \trees$.

First, assume there exist $i,j$ such that $v_i \neq v_j$.
It follows that none of $u$'s neighbors is the receiver due to the fact that the receiver can be identified as the closest node to all return addresses. 
So, $u$, not being aware of the remaining nodes and their coordinates in the system, cannot identify the receiver. 

For the second case, assume that indeed $v_i = v_j$ for all $1 \leq i,j, \leq \trees$ but there exists an $i$ such that
$cpl(hc(\id(v_i),\tilde{k}_i), y_i) < |\id(v_i)|$, i.e., the common prefix length of $\id(v_i)$ and the target coordinate is less than the length of $v_i$'s coordinate.
Then $v_i$ cannot be the receiver because at least the last element in the $i$-th coordinate of $v_i$ does not agree with $\id_i(\e_y)$.
So, again the receiver is not a neighbor of $u$ and hence $u$ is unable to identify $\e$ due to its limited view of the overlay.

Third, assume that indeed $v_i = v_j$ for all $1 \leq i,j, \leq \trees$ and $cpl(hc(\id(v_i),\tilde{k}_i), y_i) = |\id(v_i)|$ for all $i$. Then, the node $v_i$ can potentially be the receiver but so can any node $w$ that is a descendant of $v_i$ in all trees.  
Any return address vector of $w$ would result in the same results as a return address vector of $v_i$ from $u$'s local point of view.
Due to its restricted topology knowledge, $u$ is unaware if such a descendant $w$ exists, and hence can only guess that $\e_y$ is the receiver but cannot be certain.

Thus, receiver anonymity follows as the return address does not allow the unique identification of the receiver.
Sender anonymity follows analogously as a node can always forward a request from a child. 
\end{proof}

Note that the above proof does not require the application of the hash cascade. However, without the application, an attacker can always infer the common prefix length of two receiver addresses.
If we apply the hash cascade, the attacker can only determine the distance of receiver addresses to its own coordinates but might be unable to detect how close two addresses actually are. 
In this manner, the attacker can only infer very limited topology information. By hiding the topology of the social graph, we prevent the identification of users by comparing a pseudonymous topology to an external social graph.

%% file: conc.tex
\section{Conclusion}
\label{sec:conc}

We have introduced a privacy-preserving, efficient, and resilient design for F2F overlays.
For this purpose,  we have developed an algorithm for the generation of anonymous return addresses
Furthermore, we have designed multiple parallel network embeddings to enable both efficiency and resilience, as validated by an extensive simulation study.
 
Extending our simulation results, we are currently integrating our algorithms in an existing F2F overlay and have started initial testbed studies to better understand the system and its performance in real environments.

%% file: appendix.tex
\section*{Appendix}

In this section, we show how to obfuscate return addresses in VOUTE further.
PPP addresses are then generated by adding an additional layer of symmetric encryption to RAP return addresses generated by Algorithm \ref{algo:rap}. 
The idea of the approach is to allow $u$ to determine if the common prefix length of a coordinate is longer than $cpl(\pre(y),\id(u))$ but not the actual length.
For this reason, the additional layer can only be applied when using the common prefix length as a distance.
In a nutshell, we generate PPP return address through symmetrically encryption of a RAP return address using key material only known within certain subtrees.

Let $Enc: H \times \keys_{Sym} \rightarrow H$ be a semantically secure symmetric encryption function onto  $h$'s image $H$ with keyspace $\keys_{Sym}$. 
$Dec: H \times Sym \rightarrow H$ denotes the corresponding decryption.
For each subtree of the spanning tree, we distribute keys.
To achieve that, each internal node $w$ at level $l$ generates a symmetric key $k_l(w)$ by a pseudo-random key generation algorithm $SymGen$. 
Subsequently, $w$ distributes $k_l(w)$ to all its descendants. 
In this manner, a node $v$ at level $\tilde{l}$ obtains keys $k_1(v), \ldots , k_{\tilde{l}-1}(v)$ such that $k_{\lambda}(v)$ was generated by $v$'s ancestor at level $\lambda$ and forwarded to $v$ along the tree edges.
So, $k_\lambda(v)$ is known to all nodes having a common prefix length of at least $\lambda$ with $v$.
After generating a RAP return address $y=(d_1, \ldots , d_\lengthAd)$, $v$ additionally encrypts the $\lambda+1$-th element with the key
$k_\lambda(v)$, constructing the return address $y'=(d'_1, \ldots , d'_\lengthAd)$ with
\begin{align}
\label{eq:sym}
d'_j = \begin{cases}
Enc(k_{j-1}(v), d_j), & 2\leq j \leq l \\
d_j, & \textnormal{otherwise}
\end{cases}.
\end{align}
The second case in Eq. \ref{eq:sym} treats the first element, which remains unencrypted, and the randomly chosen padding, which does not agree with any coordinate.  
After generating $y'$, $v$ publishes $y'$, the routing information $\tilde{k}$ for generating $y$ and $mac(\keymac(v), y')$. 
The pseudo code of the additional encryption is displayed in Algorithm \ref{algo:ppp}.

A third realization $\routePPPCPL$ of the routing algorithm $\routeN$ is given by the construction of PPP addresses. 
During routing, a node $u$ at level $l$ first applies the decryption function to the second to $l+1$-th element of the return address $y'=(d'_1, \ldots , d'_\lengthAd)$. 
So, $v$ obtains $f(y')=(z_1, \ldots, z_{l+1})$ with
\begin{align*}
z_j = 
\begin{cases}
d'_1, & j = 1 \\
Dec(k_{j-1}(u), d'_j),  &\textnormal{ otherwise}
\end{cases}.
\end{align*}
Afterwards, $u$ determines $cpl(f(y'), cash(\tilde{k}, c))$ for all coordinates $c$ in its neighborhood. Note that $cpl(f(y'), cash(\tilde{k}, c)$ is only a lower bound on $cpl(\pre(y),c)=cpl(y,cash(\tilde{k}, c)$ because $u$ is not able to correctly decrypt some elements of $y'$.
Based on the common prefix length, $v$ can evaluate the diversity measure
\begin{align}
\label{eq:distPPP}
\distIndex{PPP-CPL}^u(y',\tilde{k},c) = \delta_{CPL}(f(y'), cash(c, \tilde{k}))
\end{align} 
for the distance $\delta_{CPL}$ defined in Eq. \ref{eq:CPLD}.
In this manner, the node $u$ obtains a set of all neighbors closer to the destination than itself.
So, $u$ chooses a random node from this set as the next hop. 

We here give a short intuition on why Algorithm \ref{algo:ppp} indeed generates PPP anonymous return addresses.
A formal proof is sketched in Section \ref{sec:voute-security}.
Let $u$ be a arbitrary node and $y'$ be a return address generated by $v$, a node at level $l_v$.
If $cpl(\id(v), \id(u))=\lambda$, $u$ correctly decrypts the first $\lambda+1$ elements of $y'$ because $k_i(u)=k_i(v)$ for $i=1\ldots \lambda$. 
Due to the semantic security of the symmetric encryption, $u$ cannot infer information about the remaining elements of $y$ from $d'_{\lambda+2}, \ldots , d'_{l_u}$ because $u$ does not know $k_i(v)$ for $i > \lambda+1$.
Thus, $y'$ indeed only reveals if a coordinate $c$ shares a longer common prefix with $\pre(y)$ than $\id(v)$.

\begin{minipage}[ht]{.9\linewidth}
\vspace{-1em}
\begin{algorithm}[H]
\caption{\small addPPPLayer()}
\label{algo:ppp}
\begin{algorithmic}[1]
\small{
\commIt{Input: RAP return address $y=(d_1, \ldots , d_\lengthAd)$}
\commIt{Internal State: Keys $k_1(v),\ldots , k_{l-1}(v)$, $Enc_{Sym}$}
\FOR{$i=2\ldots l$}
\STATE $d_j \leftarrow Enc_{Sym}(k_{j-1}(v), d_j)$ \commItLine{Encrypt element $j$}
\ENDFOR
}
\end{algorithmic}
\end{algorithm}
\end{minipage}